\documentclass[a4paper,11pt]{article} 
\usepackage[a4paper,hmargin=2.8cm,vmargin=3cm]{geometry}
\usepackage[utf8x]{inputenc} 
\usepackage[english]{babel}
\usepackage{amsmath,amsthm,amsfonts,amssymb,mathrsfs,bbm}
\usepackage{authblk}
\usepackage{lineno}
\usepackage{empheq}
\usepackage[dvipsnames]{xcolor}

\definecolor{BeauBlue}{rgb}{0, 0.2, .9}
\definecolor{BeauOrange}{rgb}{.8, .1, 0}
\usepackage{hyperref}
\hypersetup{
	colorlinks = true,
	linkcolor = BeauBlue,
	urlcolor = cyan,
	citecolor = BeauOrange,
}

\numberwithin{equation}{section}

\usepackage{tikz}
\usetikzlibrary{calc,shapes,patterns,decorations.pathreplacing}

\makeatletter
\DeclareRobustCommand\widecheck[1]{{\mathpalette\@widecheck{#1}}}
\def\@widecheck#1#2{%
    \setbox\z@\hbox{\m@th$#1#2$}%
    \setbox\tw@\hbox{\m@th$#1%
       \widehat{%
          \vrule\@width\z@\@height\ht\z@
          \vrule\@height\z@\@width\wd\z@}$}%
    \dp\tw@-\ht\z@
    \@tempdima\ht\z@ \advance\@tempdima2\ht\tw@ \divide\@tempdima\thr@@
    \setbox\tw@\hbox{%
       \raise\@tempdima\hbox{\scalebox{1}[-1]{\lower\@tempdima\box
\tw@}}}%
    {\ooalign{\box\tw@ \cr \box\z@}}}
\makeatother

\makeatletter
\newcommand{\leqnos}{\tagsleft@true\let\veqno\@@leqno}
\newcommand{\reqnos}{\tagsleft@false\let\veqno\@@eqno}
\reqnos
\makeatother

\newcommand{\triple}[1]{{\left\vert\kern-0.25ex\left\vert\kern-0.25ex\left\vert #1 
    \right\vert\kern-0.25ex\right\vert\kern-0.25ex\right\vert}}

\newcommand{\bbbone}{\mathchoice {\rm 1\mskip-4mu l} {\rm 1\mskip-4mu l}
{\rm 1\mskip-4.5mu l} {\rm 1\mskip-5mu l}}

\newtheorem{theorem}{Theorem}[section] 
 
\newtheorem{corollary}[theorem]{Corollary}
\newtheorem{lemma}[theorem]{Lemma}

\newtheorem{assumption}[theorem]{Assumption}

\title{Fluctuation Relations associated to an arbitrary bijection in path space}

\author[1]{Raphael Chetrite \thanks{raphael.chetrite@univ-cotedazur.fr}}
\author[2,3]{Stefano Marcantoni \thanks{Corresponding author: stefano.marcantoni@gssi.it}} 
\affil[1]{Institut de Physique de Nice (INPHYNI), Universit\'e C\^ote d'Azur, CNRS, \protect\\ 17 rue Julien Laupr\^etre, 06200 Nice, France}
\affil[2]{Laboratoire J. A. Dieudonn\'e (LJAD), Universit\'e C\^ote d'Azur, CNRS, \protect\\ Parc Valrose, 06108 Nice, France}
\affil[3]{Mathematics Division,  Gran Sasso Science Institute,  \protect\\ Viale Rendina 24-26-28, 67100 L'Aquila, Italy}

\begin{document}

\maketitle

\begin{center}
    \emph{Dedicated to Claudio Landim on the occasion of his 60th birthday}
\vspace{0.5cm}
\end{center}

\begin{abstract}
We introduce a framework to identify Fluctuation Relations for vector-valued observables in physical systems evolving through a stochastic dynamics. These relations arise from the particular structure of a suitable entropic functional and are induced by transformations in trajectory space that are invertible but are not involutions, typical examples being spatial rotations and translations. In doing so, we recover as particular cases results known in the literature as isometric fluctuation relations or spatial fluctuation relations and moreover we provide a recipe to find new ones. We mainly discuss two case studies, namely stochastic processes described by a canonical path probability and non degenerate diffusion processes. In both cases we provide sufficient conditions for the  fluctuation relations to hold, considering either finite time or asymptotically large times.
\end{abstract}

\tableofcontents     

\section{Introduction}

A Fluctuation Relation is a constraint that manifests when studying the statistics of a physical observable, in particular when comparing the probability that the observable attains a certain value with the probability that it attains another value under (possibly) different but closely related processes. Usually, it states that some physical outcomes are exponentially suppressed (in terms of relevant control parameters) with respect to others. For example, when the observable is the entropy production, fluctuation relations can be considered a refinement of the second law of thermodynamics, quantifying the probability of observing currents flowing in the ``wrong'' direction. Their importance also relies on the fact that they are among the few results in nonequilibrium statistical mechanics that hold true arbitrarily far from equilibrium, and other well-known results can be derived from them, like the fluctuation-dissipation theorem in the linear-response regime.

Investigations around this topic started from the numerical observation of a symmetry in the fluctuations of pressure in an externally-driven thermostatted particle system \cite{ECM93}.
The first theoretical result followed soon and it is the celebrated Gallavotti-Cohen fluctuation theorem, constraining the entropy production in stationary states of deterministic chaotic systems \cite{GC95a,GC95b}. Another influential identity for the work performed in closed systems with external driving has been later found by Jarzynski \cite{Jar97}. For these kind of results, randomness arises only from the probability distribution of the initial condition.

After the initial phase, research about fluctuation relations mainly involved systems with a stochastic dynamics, with the works \cite{Kur98,LS99,Mae99} being the seminal contributions. 
It was indeed understood (see for instance \cite{Mae99} for a detailed account) that these relations can be obtained almost tautologically by defining a generalized entropic functional as the Radon-Nikodim derivative of two probability measures of interest and using some duality relations that exploit the involution property of time-reversal (or similar involutions).
Many relations followed, using different approaches (see e.g. \cite{Sei12} for a review), so that the need for a unifying frameworks became prominent.  In the context of diffusion processes a first unification was presented in \cite{CG08} (see also \cite{Che18} for general processes), where the variety of results was linked to the different ways of defining a reversed protocol.

However, some fluctuation relations still could not fit into this scheme. These were named ``isometric'' fluctuation relations because involved the probability of observables transformed under some general isometries that do not satisfy the involution property, like spatial rotations (in fact they were also dubbed ``spatial'' later on \cite{PRG15}). They appeared for the first time in the context of Macroscopic Fluctuation Theory in \cite{Hur+11}, where the authors compared the probability of a (macroscopic) current in a chosen direction with the probability of the current rotated by an arbitrary angle, and found the exponential bias typical of fluctuation relations. Apart from follow-ups of the same authors \cite{Hur+14}, they were further studied and generalized in \cite{LG14,LG15} for equilibrium states with broken discrete and continuous symmetries, and in \cite{VHT14,VH16} for anisotropic non-equilibrium systems. Importantly, these new relations were also tested experimentally with some anisotropic rod moving on a substrate of hard spheres \cite{Kum+15} and with a hot Brownian swimmer \cite{Fal+16}. Moreover, a theoretical investigation of these relations from a microscopic point of view was later conducted in \cite{PRG15,MPG20}. In particular, in the second paper the connection with canonical biasing with respect to a given observable was highlighted, thus extending the set of candidate observables beyond currents. In \cite{VPG20}, the isometric fluctuation relations were also used to derive thermodynamic uncertainty relations for Markov processes in continuous $d$-dimensional space, thus showing their importance in deducing further constraints out of equilibrium.

In this paper, we aim at including also these fluctuation relations into a general framework, to better understand what are the assumptions needed to get those results. The known fluctuation relations from the literature are then obtained as particular cases of our main findings. We focus here on the general ideas that allow to unify different relations and frame the results into a language that is familiar to physicists, trying at the same time to be precise from a mathematical point of view. In cases where full rigor is not reached we point to the relevant mathematical literature that could be useful to proceed in the program.

The paper is structured as follows.
In Section~\ref{sec:setup}  we introduce the notation and the fundamental background for the derivation of the results. In particular, we define the generalized scalar entropic functional that is the central object around which the paper is built. We state its fundamental symmetry property that follows directly from the definition and that we dub ``mother fluctuation relation''. We also mention some basic notions of large deviation theory that are needed in the following. In Section~\ref{sec:Inv} we recall the usual case of fluctuation relations induced by involutions, and highlight the properties that do not extend to the more general case. In Section~\ref{sec:General--transformation} we state our general result and the assumptions needed for its derivation. This is the core of our unifying framework and the most important point for the applications in physics. However, the considered assumptions may be difficult to check in general. For this reason, we present in the next Sections two concrete situations where they can be verified. In particular, the so-called \emph{canonical processes}, that are processes obtained from a base process via a tilting procedure, are discussed in Section~\ref{sec:canonical}. In this case, Markovianity is not needed to obtain the result and we clarify this with an example about semi-Markov processes with finite state space. Assuming a particular form for the semi-Markov kernel and for a particular observable we are able to compute the moment generating function analytically at any time and the asymptotic cumulant generating function for large times, that allows to access also the large deviation rate function. One can then easily check by means of the explicit formulas that the fluctuation relations are satisfied. As a second example in the canonical setting we also consider a multidimensional Langevin equation with quadratic potential. The Gaussian nature of the problem allows to treat it without resorting to the overdamped limit and we can compute analytically the asymptotic cumulant generating function and the rate function. Also in this case we can check the fluctuation relations. Finally, in Section~\ref{sec:Markovian} we consider the case of generic non-degenerate diffusion processes. We present a result about finite-time fluctuation relations assuming some constraints on the drift and covariance matrix and we conjecture about the form of the respective time-asymptotic results. Even if our assumptions may seem very restrictive, we do not need to assume a constant diffusion matrix, and therefore we can generalize previous results appeared in the literature, notably \cite{PRG15}. We also discuss an example dealing with multiplicative noise to clarify this aspect.

\section{Setup and Notation}
\label{sec:setup}

In this Section we fix the notation and introduce the fundamental ingredients for the following analysis. Focusing on physical systems with a stochastic dynamics, our aim is to provide a general and unifying framework for the different fluctuation relations appeared in the literature. In turn, this general scheme can inspire and facilitate the discovery of new physically relevant relations in concrete models.

In particular, we consider stochastic processes $X_t \in E$ in the time-interval $[0,T]$ specified through a path probability measure $\mathbb{P}_{\left[0,T\right]}$, that is a probability measure on the space of paths, or trajectories, $\Omega_T$. In general, one chooses $\Omega_T= D([0,T], E)$, that is the space of c\'adl\'ag (right-continuous with left-limit) functions, but depending on the application we could also restrict to the space of continuous functions $C([0,T], E)$. If not explicitly stated, the stochastic processes considered are not supposed to be Markovian.

We also introduce $R$, a bijective transformation in the space of
trajectories $R : \Omega_T \to \Omega_T$. In many applications, $R$ can be thought of as an involution, e.g. time-reversal, but in the present manuscript we want to go beyond this restriction (the involution case will be nevertheless recalled in Section~\ref{sec:Inv}). In particular, we want to include the case of spatial rotations considered in this set-up for the first time in \cite{Hur+11}.

The main object of interest in this paper is the statistics of some $n$-dimensional real stochastic observables, that are functionals\footnote{All functionals considered in this paper are assumed to be measurable.} $A_T$ on the path space, $A_T: \Omega_T \to \mathbb{R}^n$. Often, that satisfy the time-additive property, namely $A_{T+S}= A_T + A_S \circ \theta_T$, where $\theta_T$ is time-translation by $T$. 
Example of such additive observable is 
\[
A_{T}=\int_{0}^{T}dtf(X_{t}),
\]
where $f$ is a scalar field on $E$.

Given a path measure $\mathbb{P}$, the statistics of $A_T$ is completely described by the object
\begin{equation}\label{eq:probobs}
    \mathcal{P}_{A_{T}} \equiv \mathbb{P} \circ A_T^{-1} 
\end{equation}
that is a probability measure on $\mathbb{R}^n$. Here $A_T^{-1}$ is the pre-image under the functional $A_T$, namely, given a subset $S \subseteq \mathbb{R}^n$ one has $A_T^{-1}(S) = \{ X_{\left[0,T\right]} \in \Omega_T \; | \; A_T(X_{\left[0,T\right]}) \in S \}$. This notation is common in the probabilistic literature and should not be confused with the inverse, even though the two notions coincide for bijections. 

\subsection{Entropic Functional}
\label{subsec:entropic_functional}
  
Given two different path measures $\mathbb{P}_{\left[0,T\right]}$ and $\mathbb{Q}_{\left[0,T\right]}$ on the same path space $\Omega_T$ and the transformation $R$, we define a scalar entropic functional $\varSigma_{T}^{\mathbb{P},\mathbb{Q}} : \Omega_T \to \mathbb{R}$ as follows
\begin{equation}\label{eq:AF}
\exp\left(-\varSigma_{T}^{\mathbb{P},\mathbb{Q}}\right)  \equiv\frac{d\mathbb{Q}_{\left[0,T\right]} \circ R }{d\mathbb{P}_{\left[0,T\right]}} \, ,
\end{equation}
where the notation on right-hand side denotes the Radon-Nikodym derivative. The existence of this quantity requires that $d\mathbb{Q}_{\left[0,T\right]}\circ R$ is absolutely continuous with respect to $d\mathbb{P}_{\left[0,T\right]}$. In the following, we also assume that the converse is true so that the two measures are in fact equivalent.

From this definition, denoting by $\mathbb{E}_{\mathbb{P}}$ the expectation with respect to the measure $\mathbb{P}$, we get almost directly the following lemma that we call the \emph{``Mother'' Fluctuation Relation} for an arbitrary functional of the path $Z : \Omega_T \to \mathbb{R}$.
\begin{lemma}[Mother fluctuation relation]\label{lem:mother}
Given two path measures $\mathbb{P}$ and $\mathbb{Q}$ on some trajectory space $\Omega_T$, a bijection $R : \Omega_T \to \Omega_T$ such that $R \Omega_T = \Omega_T$ and a scalar entropic functional defined as in \eqref{eq:AF}, the following relation holds 
   \begin{equation}
\mathbb{E}_{\mathbb{Q}}\left[Z\left(X_{\left[0,T\right]}\right)\right]=\ensuremath{\mathbb{E}_{\mathbb{P}}\left[\exp\left(-\varSigma_{T}^{\mathbb{P},\mathbb{Q}}\right)Z\left(R\left(X_{\left[0,T\right]}\right)\right)\right]}.\label{eq:MRF}
\end{equation} 
\end{lemma}
\begin{proof}
Starting from the definition of the expectation on the left-hand side, we find the following chain of equalities
  \begin{align}
  \mathbb{E}_{\mathbb{Q}}\left[Z\left(X_{\left[0,T\right]}\right)\right] &\equiv \int_{\Omega_T} Z\left(X_{\left[0,T\right]}\right) d\mathbb{Q}\left(X_{\left[0,T\right]}\right)  \nonumber \\
  &=  \int_{R^{-1}\Omega_T} Z\left( R \left( X_{\left[0,T\right]}\right) \right) d\mathbb{Q}\left( R\left( X_{\left[0,T\right]}\right)\right)  \nonumber \\
  &=  \int_{\Omega_T} Z\left( R \left( X_{\left[0,T\right]}\right) \right) \exp\left(-\varSigma_{T}^{\mathbb{P},\mathbb{Q}}\right) d\mathbb{P}\left(X_{\left[0,T\right]}\right) \nonumber \\
  &\equiv \mathbb{E}_{\mathbb{P}}\left[\exp\left(-\varSigma_{T}^{\mathbb{P},\mathbb{Q}}\right) Z\left(R\left(X_{\left[0,T\right]}\right)\right)\right]
  \end{align}
where in the first step we used the fact that the inverse $R^{-1}$ exists because $R$ is a bijection, and in the second step we exploited the property $R^{-1}\Omega_T = \Omega_T = R \Omega_T$ and the definition of the scalar entropic functional \eqref{eq:AF}.   
\end{proof}

\noindent
Note that this relation does not require any Markovianity and/or equilibrium assumptions. In the particular case $\mathbb{Q}=\mathbb{P}$, for simplicity we denote all the objects with a unique super-index, for instance $\varSigma_{T}^{\mathbb{P},\mathbb{P}}\rightarrow \varSigma_{T}^{\mathbb{P}}$
for the scalar entropic functional. In this case, for a functional $Z\left(X_{\left[0,T\right]}\right)$ the previous relation becomes 
\begin{equation}
\mathbb{E}_{\mathbb{P}}\left[Z\left(X_{\left[0,T\right]}\right)\right]=\mathbb{E}_{\mathbb{P}}\left[\exp\left(-\varSigma_{T}^{\mathbb{P}}\right)Z\left(R\left(X_{\left[0,T\right]}\right)\right)\right].\label{eq:MRF-1}
\end{equation}
In the following, starting with Section \ref{sec:General--transformation}, we will restrict to this case.

It is crucial to remark that with respect to the usual case \cite{Mae99,LS99,CG08,Sei12,Rol+24} where $R$ is the time-reversal, i.e an involution ($R^{2}=\mathrm{Id}$), some physically relevant properties of the entropic functional \eqref{eq:AF} are still valid, \emph{but other properties are lost}. In particular, the following properties remain valid:
\begin{enumerate}
\item Thanks to the normalisation of the path probability $\mathbb{Q}_{\left[0,T\right]}$,
the choice $Z=1$ in (\ref{eq:MRF}) gives the `Generalized
Integral Fluctuation Relation' (GIFR)
$$
\mathbb{E}_{\mathbb{P}}\left[\exp\left(-\varSigma_{T}^{\mathbb{P},\mathbb{Q}}\right)\right]=1.
$$
\item Moreover, the average value with respect to $\mathbb{P}$ of $\varSigma_{T}^{\mathbb{P},\mathbb{Q}}$ can be written as a Kullback-Leibler divergence, i.e. 
$$
\mathbb{E}_{\mathbb{P}}\left[\varSigma_{T}^{\mathbb{P},\mathbb{Q}}\right]=D_{KL}\left[\left.\mathbb{P}_{\left[0,T\right]}\right\Vert \mathbb{Q}_{\left[0,T\right]}\circ R\right].
$$
\item The last property implies that we have the `Generalized Second Law'
\[
\mathbb{E}_{\mathbb{P}}\left[\varSigma_{T}^{\mathbb{P},\mathbb{Q}}\right]\geq0.
\]
\end{enumerate}
On the other hand, the following properties that are true for involutions do not hold any more:
\begin{enumerate}
\item There is no general duality relation, i.e. a relation between
$\varSigma_{T}^{\mathbb{P},\mathbb{Q}}$ and $\varSigma_{T}^{\mathbb{Q},\mathbb{P}}$, because
$$\varSigma_{T}^{\mathbb{P},\mathbb{Q}}\circ R=-\ln\frac{d\mathbb{Q}_{\left[0,T\right]}\circ R^{2}}{d\mathbb{P}_{\left[0,T\right]}\circ R}$$
and there is no obvious way of connecting the right-hand side with $\varSigma_{T}^{\mathbb{Q},\mathbb{P}}$ if $R^2 \neq \mathrm{Id}$. 
\item The absence of such a duality property does not allow to obtain from
\eqref{eq:MRF} a Duality Fluctuation Relation between the statistics of $\varSigma_{T}^{\mathbb{Q},\mathbb{P}}$ and $\varSigma_{T}^{\mathbb{P},\mathbb{Q}}$
(see the relations \eqref{eq:GC1},\eqref{eq:GC2},\eqref{eq:GC3}
in the next Section). Even in the case $\mathbb{Q}=\mathbb{P},$ these relations are typically lost\footnote{Nevertheless, in some situations they may hold. See the discussion in Section \ref{subsec:remarks}}. The main goal of this article is to provide alternative results for more general $R$s.
\end{enumerate}

\subsection{Large Deviations}
\label{subsec:LD}

Some of the results given in the following are true at any time $T$. However, one is also interested in fluctuation relations that hold true asimptotically for large times. In fact, the latter are somewhat more relevant, in the sense that they may emerge in the limit even if not true at finite time. This makes them less fragile to perturbations.  In order to discuss this second family of relations, we need to recall a few elements about the theory of large deviations.
Informally, the large deviation principle for a family of probability measures $P_T(x), T\geq 0$, states that asymptotically in $T$ non-typical values of $x$ are suppressed exponentially. The notation commonly used to denote this fact is

\begin{equation*}
    P_T(X \approx x) \asymp \mathrm{e}^{-T I^{X}(x)}
\end{equation*}
where the function $I^{X}$ is called the rate function. In the following, most of the times we omit the superscript $X$ relative to the observable when there is no room for confusion. More precisely \cite{DZ98}, the sequence $\{ P_T \}$ of probability measures on the probability space $(\chi, \mathcal{B})$ satisfies a large deviation principle with rate function $I$ if for any Borel set $B \in \mathcal{B}$ one has
\begin{equation*}
    - \underset{x \in B^\circ}{\inf} I(x) \leq \underset{T \to \infty} {\liminf} \frac{1}{T} \ln P_T(B) \leq \underset{T \to \infty} {\limsup} \frac{1}{T} \ln P_T(B) \leq - \underset{x \in \overline{B}}{\inf} I(x) \;,
\end{equation*}
where $B^\circ$ and $\overline{B}$ denote the interior and the closure of $B$, respectively. In the framework of large deviation theory, fluctuation relations that hold true asymptotically in time translate into constraints on the rate function.

Sometimes, when studying the statistics of some random functional (one-dimensional here for simplicity) $A_T \equiv A(X_{[0,T]})$, instead of focusing on the sequence of measures $\mathcal{P}_{A_T}$, it is more convenient to analyze the sequence of cumulant generating functions $\Lambda_T$ 
\begin{equation*}
\Lambda_T(k)  \equiv \ln \mathbb{E}_{\mathbb{P}} \left[ \exp\left(k A_T\right) \right] \;.
\end{equation*}

In this respect, a very useful result is the G\"artner-Ellis theorem \cite{DZ98} according to which the large deviation principle for the sequance of measures $\{ \mathcal{P}_{A_T/T} \}$ holds if the scaled cumulant generating function $\Lambda(k)$, defined as the following limit 
$$\Lambda (k) = \lim_{T \to \infty} \frac{\Lambda_T (k)}{ T},$$
is finite and differentiable for any $k \in \mathbb{R}$. Moreover, the rate function $I(a)$ is the Legendre-Fenchel transform of $\Lambda$, namely
\begin{equation*}
   I(a) = \sup_{k \in \mathbb{R}} \big( k a - \Lambda(k) \big) \;,
\end{equation*}
and therefore it is strictly convex. As a result, typically, if $\Lambda (k)$ has a symmetry, a corresponding symmetry is present for the rate function. We will see examples of this in the next Sections. 

More generally, we are interested in vector-valued observables $A_T \in \mathbb{R}^n$, whose cumulant generating function reads as before, with $k A_T$ replaced by $\langle  k , A_T\rangle$, for $k \in \mathbb{R}^n$. Moreover, the function $\Lambda(k)$ may only be finite on a subset $S \subset \mathbb{R}^n$. In this case, the G\"artner-Ellis theorem still applies provided that $\Lambda(k)$ is \emph{essentially smooth} \cite{DZ98}, namely the following conditions hold
\begin{itemize}
    \item[(i)] the origin is in the interior of $S$, i.e. $0 \in S^\circ$
    \item[(ii)] $\Lambda(k)$ is differentiable on $S^\circ$  
    \item[(iii)] $\Lambda(k)$ is steep, in the sense that $\lim_{n \to \infty} |\nabla \Lambda(k_n)| = \infty$ for any sequence $\{ k_n\}$ in $S^\circ$ converging to a boundary point of $S^\circ$
\end{itemize}

Note that the validity of G\"artner-Ellis theorem is not necessary to have a large deviation principle on the sequence of probability measures describing the relevant observable. In particular, the rate function does not need to be convex.  However, we decided here to keep the technical details to a minimum and restrict ourselves to the simplest conditions that allow to prove interesting fluctuation relations, that are the main focus of the paper.

\section{Usual Fluctuation Relations induced by an involution}
\label{sec:Inv}

Before stating the main results of the paper, in this Section, we informally recall the usual relevant case \cite{Mae99,LS99,CG08,Sei12,Rol+24} where $R$ is an involution i.e. $R^{2}=\mathrm{Id}$.
A prototypical example is the time reversal $\left(R\left(X_{\left[0,T\right]}\right)\right)_{t}=X_{T-t}$. As previously mentioned, in such a case a duality relation exists between $\varSigma_{T}^{\mathbb{P},\mathbb{Q}}$ and $\varSigma_{T}^{\mathbb{Q},\mathbb{P}}$, namely
\begin{equation}\label{eq:duality}
    \varSigma_{T}^{\mathbb{P},\mathbb{Q}}\circ R=-\ln\frac{d\mathbb{Q}_{\left[0,T\right]}}{d\mathbb{P}_{\left[0,T\right]}\circ R}=-\varSigma_{T}^{\mathbb{Q},\mathbb{P}}.
\end{equation}
This relation, together with a suitable choice of the functional $Z$ in Eq. \eqref{eq:MRF}, can be used to generate fluctuation relations for the statistics of the scalar entropic functionals $\varSigma_{T}^{\mathbb{P},\mathbb{Q}}$ and $\varSigma_{T}^{\mathbb{Q},\mathbb{P}}$. 
In particular, one defines the law of the stochastic variable $\varSigma_{T}^{\mathbb{P},\mathbb{Q}}$ under the process $\mathbb{P}$ and the law of $\varSigma_{T}^{\mathbb{Q},\mathbb{P}}$ under the process $\mathbb{Q}$ as follows
\begin{equation}
    \mathcal{P}_{\varSigma_{T}^{\mathbb{P},\mathbb{Q}}}^\mathbb{P} \equiv \mathbb{P} \circ (\varSigma_{T}^{\mathbb{P},\mathbb{Q}})^{-1} \,, \qquad  \mathcal{P}_{\varSigma_{T}^{\mathbb{Q},\mathbb{P}}}^\mathbb{Q} \equiv \mathbb{Q} \circ (\varSigma_{T}^{\mathbb{Q},\mathbb{P}})^{-1} \;.
\end{equation}
Then, setting $Z=f(\varSigma_{T}^{\mathbb{Q},\mathbb{P}})$ in \eqref{eq:MRF} and using $\varSigma_{T}^{\mathbb{Q},\mathbb{P}} \circ R = - \varSigma_{T}^{\mathbb{P},\mathbb{Q}}$ (see \eqref{eq:duality}), one finds
\begin{equation}
  \mathbb{E}_{\mathbb{Q}}\left[ f(\varSigma_{T}^{\mathbb{Q},\mathbb{P}}) \right] = \mathbb{E}_{\mathbb{P}}\left[ \exp\left(-\varSigma_{T}^{\mathbb{P},\mathbb{Q}}\right) f(-\varSigma_{T}^{\mathbb{P},\mathbb{Q}}) \right]   \;,
\end{equation}
that can be equivalently written
\begin{equation}
    \int_{\mathbb{R}} d\mathcal{P}_{\varSigma_{T}^{\mathbb{Q},\mathbb{P}}}^\mathbb{Q} (\sigma) f(\sigma) = \int_{\mathbb{R}} d\mathcal{P}_{-\varSigma_{T}^{\mathbb{P},\mathbb{Q}}}^\mathbb{P} (\sigma) \exp(\sigma) f(\sigma)
\end{equation}
Due to the arbitrariness of $f$ the previous relation can be recast as constraint on the Radon-Nikodym derivative
\begin{equation}
    \frac{d\mathcal{P}_{\varSigma_{T}^{\mathbb{Q},\mathbb{P}}}^\mathbb{Q}}{d\mathcal{P}_{-\varSigma_{T}^{\mathbb{P},\mathbb{Q}}}^\mathbb{P}} (\sigma) = \exp(\sigma)
\end{equation}
If the two probabilities are absolutely continuous with respect to the Lebesgue measure and therefore admit a density, the previous relation can be recast into a relation between the densities
\begin{equation}
\rho_{\varSigma_{T}^{\mathbb{Q},\mathbb{P}}}^{\mathbb{Q}}\left[\sigma\right]=\exp\left(\sigma\right)\rho_{\varSigma_{T}^{\mathbb{P},\mathbb{Q}}}^{\mathbb{P}}\left[-\sigma\right].
\label{eq:DFR}
\end{equation}

 In a similar way, setting $Z=\exp\left(k\varSigma_{T}^{\mathbb{Q},\mathbb{P}}  \right)$
in \eqref{eq:MRF} we obtain a finite-time fluctuation relation connecting the moment generating functions of $\varSigma_{T}^{\mathbb{Q},\mathbb{P}}$ and $\varSigma_{T}^{\mathbb{P},\mathbb{Q}}$:
\begin{equation}\label{eq:MRF-2}
\mathbb{E}_{\mathbb{Q}}\left[\exp\left(k\varSigma_{T}^{\mathbb{Q},\mathbb{P}}\right)\right]=\mathbb{E}_{\mathbb{P}}\left[\exp\left(-\left(k+1\right)\varSigma_{T}^{\mathbb{P},\mathbb{Q}}\right)\right] \,.
\end{equation}

Moreover, if the two sequences of probability measures parametrized by $T$, $ \mathcal{P}_{\varSigma_{T}^{\mathbb{Q},\mathbb{P}}/T}^\mathbb{Q}$  and $\mathcal{P}_{-\varSigma_{T}^{\mathbb{P},\mathbb{Q}}/T}^\mathbb{P} $, obey a large deviation principle with rate $T$ and rate functions $I_{\mathbb{Q},\mathbb{P}}$ and $I_{\mathbb{P},\mathbb{Q}}$, respectively
\begin{equation}\label{eq:ldestimates}
  \begin{cases}
\mathbb{Q} \left(\varSigma_{T}^{\mathbb{Q},\mathbb{P}} /T\approx \sigma \right) \asymp\exp\left(-TI_{\mathbb{P},\mathbb{Q}}(\sigma)\right)\\
\mathbb{P} \left(\varSigma_{T}^{\mathbb{P},\mathbb{Q}} /T\approx \sigma \right)  \asymp\exp\left(-TI_{\mathbb{Q},\mathbb{P}} (\sigma)\right)
\end{cases}  ,
\end{equation}
one obtains the \emph{duality asymptotic Fluctuation
Relation}
\begin{equation}
I_{\mathbb{Q},\mathbb{P}}\left(\sigma\right)=-\sigma+I_{\mathbb{P},\mathbb{Q}}\left(-\sigma\right).\label{eq:GC1}
\end{equation}

Also, if one can show that asymptotically in time the moment generating functions behave as follows,
\begin{equation}\label{eq:ldlaplace}
 \begin{cases}
\mathbb{E}_{\mathbb{P}}\left[\exp\left(k\varSigma_{T}^{\mathbb{P},\mathbb{Q}}\right)\right]\asymp\exp\left(T\Lambda_{\mathbb{P},\mathbb{Q}}(k)\right)\\
\mathbb{E}_{\mathbb{Q}}\left[\exp\left(k\varSigma_{T}^{\mathbb{Q},\mathbb{P}}\right)\right]\asymp\exp\left(T\Lambda_{\mathbb{Q},\mathbb{P}}(k)\right)
\end{cases},  
\end{equation}
one obtains the \emph{duality asymptotic Fluctuation
Relation} in the Laplace version 
\begin{equation}
\Lambda_{\mathbb{Q},\mathbb{P}}(k)=\Lambda_{\mathbb{P},\mathbb{Q}}(-k-1).\label{eq:GC2}
\end{equation}

In the particular case $\mathbb{Q}=\mathbb{P}$, the relations \eqref{eq:GC1},\eqref{eq:GC2}
become the famous \emph{asymptotic Fluctuation
Relation} for $\varSigma_{T}^{\mathbb{P}}$ \cite{GC95a,Kur98,Gal99,LS99,Mae99,CG08,Sei12,Rol+24} 
\begin{equation}
\begin{cases}
I_{\mathbb{P}}\left(\sigma\right)=-\sigma+I_{\mathbb{P}}\left(-\sigma\right)\\
\Lambda_{\mathbb{P}}(k)=\Lambda_{\mathbb{P}}(-k-1)
\end{cases}.\label{eq:GC3}
\end{equation}

On the mathematical side, the challenge is to rigorously prove the large deviation estimates \eqref{eq:ldestimates} and \eqref{eq:ldlaplace}, especially in the case of non-compact state space. See for instance \cite{BdG15,BGL23,Raq24} for recent results.

\section{General result: Fluctuation Relations (FRs) associated to an arbitrary bijection}
\label{sec:General--transformation}

\subsection{Fluctuation Relations}

In this Section, we show that even in absence of a 
Fluctuation Relation of type~\eqref{eq:GC3} for the scalar entropic functional $\varSigma_{T}^{\mathbb{P}}$~\eqref{eq:AF},
 when $R$ is not an involution, we can anyway
obtain different Fluctuation Relations under the following hypothesis. 
\begin{assumption}[Decomposition of the entropic functional and covariance hypothesis]\label{ass}
There exist (i) a vector-valued functional $A_T$ on the space of trajectories $A_T : \Omega_T \to \mathbb{R}^n$, (ii) a space-time homogeneous vector $w \in \mathbb{R}^n$, and (iii) a space-time homogeneous invertible matrix $\mathcal{R} \in M_n(\mathbb{R})$ such that
\begin{equation}\label{eq:ass1}
   \varSigma_{T}^{\mathbb{P}}=\left\langle w, A_T\right\rangle \; , \quad \text{with}  \quad A_T\left(R\left(X_{\left[0,T\right]}\right)\right)=\mathcal{R}A_T\left(X_{\left[0,T\right]}\right) \;,
\end{equation}
where the bracket $ \left\langle \cdot \, ,\cdot \right\rangle$ denotes the canonical
scalar product in $\mathbb{R}^{n}$.
\end{assumption}

We call the condition on the transformation of $A_T$ \emph{covariance of the observable}.
Note that for a generic path probability $\mathbb{P}$, a triple $(A_T, w, \mathcal{R})$ satisfying Eq.~\eqref{eq:ass1} may not exist. In particular, while writing the entropic functional as the scalar product of two vector-valued quantities can be done for free, requiring the covariance of $A_T$ with respect to $R$ is a highly nontrivial constraint on $\varSigma_{T}^{\mathbb{P}}$. Moreover, when a triple exists it is not necessarily unique, because for instance a space-time homogeneous orthogonal matrix $O$ commuting with $\mathcal{R}$ would be sufficient to construct an alternative triple $( O A_T, O w, \mathcal{R})$. 
Note finally that considering an $n$-dimensional vector observable $A_T$ cannot be easily dismissed in general. Indeed, a decomposition like \eqref{eq:ass1} for $n=1$ corresponds to the assumption that the entropic functional itself has the covariance property with respect to $R$, i.e. 
$\varSigma_{T}^{\mathbb{P}}\circ R=\mathcal{R}\varSigma_{T}^{\mathbb{P}}$ for some number $\mathcal{R}$. From the definition of $\varSigma_{T}^{\mathbb{P}}$ (\eqref{eq:AF} in the case $\mathbb{P}=\mathbb{Q}$) this relation can be rewritten as
$$\ln\left(\frac{d\mathbb{P}_{\left[0,T\right]}\circ R^{2}}{d\mathbb{P}_{\left[0,T\right]}\circ
R}\right)=\mathcal{R}\ln\left(\frac{d\mathbb{P}_{\left[0,T\right]}\circ R}{d\mathbb{P}_{\left[0,T\right]}}\right)$$
and it is often false for arbitrary $R$.
As expected, the relation is instead true for an involution $R^{2}=\mathrm{Id}$, where we find $\mathcal{R}=-1.$

In any case, assuming the validity of Assumption~\ref{ass} one can readily prove nontrivial Fluctuation Relations at finite time for the statistics of $A_T$. 
 
\begin{theorem}[Finite-time FRs]\label{thm:result1}
Under Assumption~\ref{ass} the following Finite-Time Fluctuation Relations hold true, for any $k,a \in \mathbb{R}^n$:
  \begin{equation}\label{eq:beau}
\frac{d\mathcal{P}_{A_T}}{d\mathcal{P}_{\mathcal{R}A_T}}(a)=\exp \left(-\left\langle w,\mathcal{R}^{-1} a\right\rangle \right) \; ,
\end{equation}
    \begin{equation}\label{eq:MRF-2-1-1}
\mathbb{E}_{\mathbb{P}}\Big[\exp\Big\langle k,A_T\Big\rangle \Big]=\mathbb{E}_{\mathbb{P}}\Big[\exp \Big\langle \mathcal{R}^{\dagger} k -
w,A_T\Big\rangle \Big] ,
    \end{equation}
where $\mathcal{P}_{A_T}$ is the law of the stochastic variable $A_T$.
\end{theorem}

\begin{proof}

Consider a measurable function $f: \mathbb{R}^n \to \mathbb{R}^n$. One has by definition
\begin{equation}
    \int_{\mathbb{R}^n} d\mathcal{P}_{A_T}(a) f(a) = \int_{\Omega} d\mathbb{P}(X_{\left[0,T\right]}) f(A_T(X_{\left[0,T\right]})) = \mathbb{E}_{\mathbb{P}}\left[f(A_T)\right] .
\end{equation}
Setting $Z\left(X_{\left[0,T\right]}\right)= f(A_T(X_{\left[0,T\right]})) $ 
in~(\ref{eq:MRF-1}), and using the covariance hypothesis \eqref{eq:ass1}, we obtain
\begin{align}\label{eq:MRF-1-1}
\mathbb{E}_{\mathbb{P}}\left[f(A_T)\right] &= \mathbb{E}_{\mathbb{P}}\left[\exp\left(-\left\langle w ,A_T \right\rangle \right) f(\mathcal{R}A_T)\right] \nonumber \\
&=  \int_{\mathbb{R}^n} d\mathcal{P}_{A_T}(a) \exp\langle -w, a \rangle f(\mathcal{R} a) \nonumber \\
&=  \int_{\mathbb{R}^n} d\mathcal{P}_{A_T}(\mathcal{R}^{-1} a) \exp\langle -w, \mathcal{R}^{-1} a \rangle f(a) ,
\end{align}
where the existence of $\mathcal{R}^{-1}$ is inherited from the existence of $R^{-1}$. Also, one has $\mathcal{R}\mathbb{R}^n = \mathbb{R}^n = \mathcal{R}^{-1}\mathbb{R}^n$ that is used in the last equality. More in general, the set of possible values of $A_T$ is preserved by $\mathcal{R}$, as a consequence of the covariance assumption and the hypothesis that $R\Omega_T=\Omega_T$ (see Lemma \ref{lem:mother}). Indeed, if $A_T(X_{\left[0,T\right]}) = a$ for some path $X_{\left[0,T\right]}$, then $Y_{\left[0,T\right]} \equiv R^{-1} X_{\left[0,T\right]}$ is another allowed path, because $R\Omega_T=\Omega_T=R^{-1}\Omega_T$, and the covariance assumption \eqref{eq:ass1} ensures that $\mathcal{R}A_T(Y_{\left[0,T\right]}) = a$, so that there is indeed a path where $\mathcal{R}A_T= a$.
Recalling that $\mathcal{P}_{A_T} \circ \mathcal{R}^{-1}= \mathcal{P}_{\mathcal{R}A_T} $, one has the equality
\begin{equation}
    \int_{\mathbb{R}^n} d\mathcal{P}_{A_T}(a) f(a) = \int_{\mathbb{R}^n} d\mathcal{P}_{\mathcal{R}A_T}(a) \exp\langle -w, \mathcal{R}^{-1} a \rangle f(a) ,
\end{equation}
and the arbitrariness of $f$ implies the validity of \eqref{eq:beau}.

The symmetry \eqref{eq:MRF-2-1-1} on the moment generating function of $A_T$, can be obtained setting $Z\left(X_{\left[0,T\right]}\right)=\exp\left(\left\langle k ,A\left(X_{\left[0,T\right]}\right)\right\rangle \right)$
in \eqref{eq:MRF-1}. Explicitly,
\begin{equation}\label{eq:MRF-2-1}
\mathbb{E}_{\mathbb{P}}\Big[\exp\left\langle k ,A_T\right\rangle \Big]=\mathbb{E}_{\mathbb{P}}\Big[\exp\left(-\left\langle w,A_T\right\rangle \right)\exp\left\langle k ,\mathcal{R}A_T\right\rangle \Big],
\end{equation}
and this can then be easily recast as Eq.~\eqref{eq:MRF-2-1-1} using $\left\langle k ,\mathcal{R}A_T\right\rangle = \left\langle \mathcal{R}^\dag k ,A_T\right\rangle$.
\end{proof}

Theorem \ref{thm:result1} is the first important result of this paper and it essentially follows just from the decomposition \eqref{eq:ass1} of the entropic functional. At any finite time $T$, it gives a precise comparison between the probability measure of the vector-valued observable $A_T$ and the probability measure of the transformed observable $\mathcal{R}A_T$. Moreover, through the symmetry of the moment generating function (and consequently of the cumulant generating function) it encodes the hidden constraints on the statistics of $A_T$ (see also Section \ref{subsec:cumulants} for more discussions in this direction).

The relations \eqref{eq:beau} and \eqref{eq:MRF-2-1-1} can translate in the large deviation regime of $A_T$ if it is well-defined. Informally, one could say that
  \begin{equation}
\begin{cases}\label{eq:Ldev}
\mathbb{P} \left( A_T /T \approx a\right) \asymp \exp\left(-TI_{\mathbb{P}} ( a )\right)\\
\mathbb{E}_{\mathbb{P}}\left[\exp\left(\left\langle k,A_T\right\rangle \right)\right]\asymp\exp\left(T\Lambda_{\mathbb{P}}(k)\right)
\end{cases}.
\end{equation}
More precisely, we want to consider the following Assumptions.
\begin{assumption}[Weaker decomposition]\label{ass_weak}
  There exists a triple $(A_T, w, \mathcal{R})$ as in Assumption \ref{ass}, with $A_T$ satisfying \eqref{eq:ass1}, and such that 
  \begin{equation}\label{eq:weakdec}
       \varSigma_{T}^{\mathbb{P}}=\left\langle w, A_T\right\rangle + B_T \; ,
  \end{equation}
  where $B_T$ is a term subleading in $T$, i.e. almost surely
  \begin{equation}\label{eq:subleading}
     \lim_{T\to \infty}\sup_{X_{\left[0,T\right]}}|B_T|/T=0 .
  \end{equation}
\end{assumption}
\begin{assumption}[Large deviation principle]\label{ass:ld}
    Let $A_T: \Omega_T \to \mathbb{R}^n$ be a vector-valued functional in path space. For any $k \in \mathbb{R}^n$ the following limit exists and is finite
   \begin{equation}\label{eq:cumul}
      \Lambda_{\mathbb{P}}(k) = \lim_{T \to \infty} \frac{1}{T}\ln \mathbb{E}_\mathbb{P} \left[\exp\left\langle k,A_T\right\rangle \right] ,
   \end{equation}
   and the function $\Lambda_{\mathbb{P}}(k)$ is differentiable\footnote{An elementary counter-example of this assumption comes for the scalar random variable $A_{T}=T\left(2\mathcal{B}_{\frac{1}{2}}-1\right)$, 
with $\mathcal{B}_{\frac{1}{2}}$ the Bernoulli random variable with
parameter $\frac{1}{2}$, where one can directly show that $ \Lambda_{\mathbb{P}}(k)=\left|k\right|$.}. Alternatively, the limit is finite only on a set $S \subset \mathbb{R}^n$ and the the function $\Lambda(k)$ is essentially smooth\footnote{Recall the definition given in Section \ref{subsec:LD}}.
\end{assumption}

Given these assumptions, that need to be proved in more specific mathematical models describing physically relevant contexts, one arrives at the next Theorem that is one of the main results of this article. 
\begin{theorem}[Asymptotic FRs] \label{thm:result2} 
    If both Assumption~\ref{ass_weak} and Assumption~\ref{ass:ld} are verified, the following Asymptotic Fluctuation Relation hold for the scaled cumulant generating function of the observable $A_T$ 
    \begin{equation}\label{eq:symlaplace}
    \Lambda_{\mathbb{P}}\big(\big(\mathcal{R}^\dag\big)^{-1} k \big)=\Lambda_{\mathbb{P}}\big(k - w\big) . 
    \end{equation}
    Moreover, a large deviation principle exists for the family of probability measures $\mathcal{P}_{A_T/T}$ with speed $T$ and rate function $I_{\mathbb{P}}$ and the rate function satisfies the Asymptotic Fluctuation Relation
    \begin{equation}\label{eq:MR-4}
      I_{\mathbb{P}}\left(\mathcal{R} a\right)=\left\langle w,a \right\rangle +I_{\mathbb{P}}\left(a\right)  .
    \end{equation}
\end{theorem}

\begin{proof}
Substituting $Z= \exp\left\langle k,A\left(X_{\left[0,T\right]}\right)\right\rangle $ in \eqref{eq:MRF-1} and using the decomposition \eqref{eq:weakdec} for the scalar entropic functional one finds the equality
\begin{equation}
    \mathbb{E}_{\mathbb{P}}\left[\mathrm{e}^{\left\langle k ,\,A_T\right\rangle }\right]=\mathbb{E}_{\mathbb{P}}\left[\mathrm{e}^{-\left\langle w, \,A_T\right\rangle  - B_T}\,\mathrm{e}^{\left\langle k ,\, \mathcal{R}A_T\right\rangle  }\right],
\end{equation}
valid for any $T\geq 0$. The right-hand side is bounded as follows
\begin{equation}
    \mathbb{E}_{\mathbb{P}}\left[\mathrm{e}^{\left\langle \mathcal{R}^\dag k - w, \, A_T \right\rangle }\right] \mathrm{e}^{-\sup_{X_{\left[0,T\right]}}|B_T|} \leq \mathbb{E}_{\mathbb{P}}\left[\mathrm{e}^{\left\langle \mathcal{R}^\dag k -w, \, A_T \right\rangle  - B_T}\right] \leq \mathbb{E}_{\mathbb{P}}\left[\mathrm{e}^{\left\langle \mathcal{R}^\dag k - w, \,A_T \right\rangle }\right] \mathrm{e}^{\sup_{X_{\left[0,T\right]}}|B_T|} 
\end{equation}
and the symmetry \eqref{eq:symlaplace} emerges when taking the limit $T \to \infty$ of the logarithm divided by $T$, thanks to the condition \eqref{eq:subleading} in Assumption \ref{ass_weak}. The large deviation principle for the law of the observable $A_T/T$ is then obtained by means of the G\"artner-Ellis theorem (see e.g. \cite{DZ98}) thanks to the properties of $\Lambda_{\mathbb{P}}(k)$ given in Assumption \ref{ass:ld}.
\end{proof}

The G\"artner-Ellis theorem is a useful tool that allows us to derive the symmetry on the rate function without assuming \emph{a priori} the first of \eqref{eq:Ldev}. The same could be used (and has been used in the literature) in the involution case, for instance to derive \eqref{eq:GC1} without assuming \eqref{eq:ldestimates} in a ad-hoc way.

We have therefore established the asymptotic fluctuation relations \eqref{eq:symlaplace}-\eqref{eq:MR-4}, that may be rewritten also
 \begin{equation}
\begin{cases}\label{eq:MR}
 I_{\mathbb{P}}\left(a \right)=\left\langle w,\mathcal{R}^{-1} a \right\rangle +I_{\mathbb{P}}\left(\mathcal{R}^{-1} a \right)\\
 \Lambda_{\mathbb{P}}\left( k \right)=\Lambda_{\mathbb{P}}\left(\mathcal{R}^{\dagger} k -w\right) 
\end{cases}.
\end{equation}
Despite formal similarities with previous results appeared in the literature, at the present level of generality this is a new result. 
The key point is that the scalar entropic functional can be rewritten in terms of vector-valued observables that are (covariant) (in the sense of \eqref{eq:ass1}) with some bijection $R$ in trajectory space. This feature in turn hints at some underlying symmetry of the process $\mathbb{P}$ with respect to $R$, even though this is not explicitly requested in the general formalism. 

The following observation further highlights this aspect: if Assumption \ref{ass} is satisfied for the entropic functional constructed via the bijection $R$ with the triple $(A_T,w,\mathcal{R})$, then it is also satisfied for the functional constructed via the bijection $R^2$ with the different triple $(A_T,w +\mathcal{R}^\dag w,\,\mathcal{R}^2)$. Indeed, from the definition of the scalar entropic functional, after substituting the decomposition in terms of $A_T$ one has
\begin{equation*}
     \exp\big(-\left\langle w, A_T\right\rangle\big)  =\frac{d\mathbb{P}_{\left[0,T\right]} \circ R }{d\mathbb{P}_{\left[0,T\right]}} \,,
\end{equation*}
and therefore
\begin{equation*}
     \exp\big(-\left\langle w, A_T \circ R\right\rangle\big)  =\frac{d\mathbb{P}_{\left[0,T\right]} \circ R^2 }{d\mathbb{P}_{\left[0,T\right]}\circ R}  \,.
\end{equation*}
Using the covariance property of $A_T$ with respect to $R$, the left-hand side can be conveniently rewritten
\begin{equation*}
     \exp\left(-\left\langle \mathcal{R}^\dag w, A_T \right\rangle\right)  =\frac{d\mathbb{P}_{\left[0,T\right]} \circ R^2 }{d\mathbb{P}_{\left[0,T\right]}\circ R}  \,,
\end{equation*}
so that putting together the first and the third equation we end up with
\begin{equation}
     \exp\left(-\left\langle w + \mathcal{R}^\dag w, A_T \right\rangle\right)  =\frac{d\mathbb{P}_{\left[0,T\right]} \circ R^2 }{d\mathbb{P}_{\left[0,T\right]}}  \,,
\end{equation}
that on the right-hand side is the definition of the entropic functional induced by $R^2$.
Moreover, the covariance property of $A_T$ with respect to $R$ immediately implies the same property with respect to $R^2$, via the matrix $\mathcal{R}^2$, i.e.
\begin{equation*}   A_T\left(R\left(X_{\left[0,T\right]}\right)\right)=\mathcal{R}A_T\left(X_{\left[0,T\right]}\right) \quad  \Rightarrow  \quad A_T\left(R^2\left(X_{\left[0,T\right]}\right)\right)=\mathcal{R}^2A_T\left(X_{\left[0,T\right]}\right) .
\end{equation*}
This finding can be easily generalized to arbitrary powers $m \in \mathbb{N}, m \geq 1$ (multiple composition), so that, labeling $w_R$ the vector associated to the original bijection, one has the decomposition \eqref{eq:ass1} for the entropic functional induced by $R^m$ with the triple $(A_T,w_{R^m},\mathcal{R}^m)$, where explicitly
\begin{equation}
    w_{R^m}= w_R + \mathcal{R}^\dag w_{R^{m-1}} = \sum_{j=0}^{m-1}(\mathcal{R}^\dag)^j w_R \,.
\end{equation}
Note that for an involution $w_{R^2}=w_{\mathrm{Id}}=0$ so that the previous constraint immediately gives\footnote{See also the discussion in Section \ref{subsec:remarks}} $\mathcal{R}^\dag w_R = - w_R$.

\subsection{Relations between cumulants}  
\label{subsec:cumulants}

We can further explore the consequences of the Theorems \ref{thm:result1} and \ref{thm:result2}, having in mind possible constraints on the cumulants of the observable, or on the response functions, like the Onsager reciprocal relations \cite{dGM84,Str92,Str94}.
For instance, concerning hierarchical constraints on the cumulants, one can consider the case where the bijection $R$ is in fact a one-parameter family of transformations $R(\theta)$. This defines a family of scalar entropic functionals $\varSigma_{T}^{\mathbb{P},\theta}$ and one can assume that all of them satisfy Eq.\eqref{eq:ass1} with the same\footnote{This once again reflects a hidden symmetry in the problem.} observable $A_T$, that obeys the covariance condition with a matrix $\mathcal{R}(\theta)$, and a different vector $w(\theta)$. Focusing on the finite-time results, a consequence of Theorem \ref{thm:result1} is that the cumulant generating function of $A_T$, $\Lambda_{\mathbb{P},T}(k)= \ln \mathbb{E}_\mathbb{P} \left[\exp\left\langle k,A_T\right\rangle \right] $, has the symmetry $\Lambda_{\mathbb{P},T}(k)= \Lambda_{\mathbb{P},T}(\mathcal{R}^\dag (\theta) k -w(\theta)) $ for any $\theta$ and any time $T$.
This fact naturally produces constraints on the cumulants of the observable. Indeed, one can write for any integer $p$, any set of indices $(i_1,\ldots, i_p) \in \{1,2,\ldots,n \}^{p}$ and any $\theta$,
\begin{equation}
    \partial_{k_{i_1}}\ldots \partial_{k_{i_p}}\Lambda_{\mathbb{P},T} \Big|_{0} = \sum_{j_{1}\ldots j_{p}\in\{1,2,\ldots,n \}^{p}} \mathcal{R}^\dag_{j_1 i_1}(\theta) \ldots \mathcal{R}^\dag_{j_p i_p}(\theta) \,  \partial_{\eta_{j_1}}\ldots \partial_{\eta_{j_p}} \Lambda_{\mathbb{P},T} \Big|_{-w(\theta)} .
\end{equation}
Assume that $R({\theta=0})$ is the identity, so that one has also  $ \mathcal{R}(0)=\bbbone_n$,  $\varSigma_T^{\mathbb{P},\theta=0}=0$ and $w({\theta=0})=0$. 
Therefore, since the left-hand side does not depend on $\theta$, it turns out that: 
\begin{align}
 &\sum_{j_{1}\ldots j_{p}\in\{1,2,\ldots,n \}^{p}} \partial_\theta \Big( \mathcal{R}^\dag_{j_1 i_1}(\theta) \ldots \mathcal{R}^\dag_{j_p i_p}(\theta) \Big) \partial_{\eta_{j_1}}\ldots \partial_{\eta_{j_p}} \Lambda_{\mathbb{P},T} \Big|_{-w(\theta)} \nonumber \\
 &\quad - \sum_{j_{1}\ldots j_{p},j_{p+1}\in\{1,2,\ldots,n \}^{p+1}} \mathcal{R}^\dag_{j_1 i_1}(\theta) \ldots \mathcal{R}^\dag_{j_p i_p}(\theta) \, \partial_{\eta_{j_1}}\ldots \partial_{\eta_{j_p}} \partial_{\eta_{j_{p+1}}} \Lambda_{\mathbb{P},T} \Big|_{-w(\theta)} \partial_\theta (w(\theta))_{ j_{p+1}}  =0.
\end{align}
Considering the case $\theta=0$ in the previous formula, and denoting $\mathbb{E}_{\mathbb{P}}^c(A_{T,j_1}\ldots A_{T,j_p})$ the cumulants of the observable $A_T$, we get the relation\footnote{The discussed relations hold true when all the functions involved are regular enough so that all the necessary derivatives exist. Note however that this may not be the case (see e.g. \cite{Por10})}
\begin{align}
 & \sum_{j_{1}\ldots j_{p}\in\{1,2,\ldots,n \}^{p}} \partial_\theta \Big( \mathcal{R}^\dag_{j_1 i_1}(\theta) \ldots \mathcal{R}^\dag_{j_p i_p}(\theta) \Big) \Big|_{0}\mathbb{E}_{\mathbb{P}}^c(A_{T,j_1}\ldots A_{T,j_p}) \nonumber \\
 &\quad - \sum_{j_{1}\ldots j_{p},j_{p+1}\in\{1,2,\ldots,n \}^{p+1}}  \mathcal{R}^\dag_{j_1 i_1}(0) \ldots \mathcal{R}^\dag_{j_n i_n}(0) \partial_\theta (w(\theta)_{ j_{p+1}})\Big|_{0} \mathbb{E}_{\mathbb{P}}^c(A_{T,j_1}\ldots A_{T,j_p}A_{T,j_{p+1}}) =0.
\end{align}

Now, considering $\mathcal{R}(\theta)=\bbbone_n +\theta\mathcal{N}+o\left(\theta\right)$,
with $\mathcal{N}$  the infinitesimal generator of the one-parameter family of transformations,
we obtain a relation between the cumulants of order $p$ and the cumulants of order $p+1$: 
\begin{align}
\sum_{k=1}^{p}\sum_{j_{k}=1}^{n}\mathcal{N}_{i_{k}j_{k}}&\mathbb{E}_{\mathbb{P}}^{c}\left(A_{T,i_{1}}A_{T,i_{2}}...A_{T,i_{k-1}}A_{T,j_{k}}A_{T,i_{k+1}}...A_{T,i_{p}}\right) \nonumber \\
&=\sum_{j_{p+1}=1}^{n}\left.\partial_{\theta}\left(w(\theta)_{j_{p+1}}\right)\right|_{\theta=0}\mathbb{E}_{\mathbb{P}}^{c}\left(A_{T,i_{1}}A_{T,i_{2}}...A_{T,i_{p}}A_{T,j_{p+1}}\right),
\label{eq:BeauG}
\end{align} 
for any integer $p$ and any set of indices $(i_1,\ldots, i_p) \in \{1,2,\ldots,n \}^{p}$.  In particular, the case $p=1$ gives the relations 
\begin{equation}
\sum_{j_{1}=1}^{n}\mathcal{N}_{i_{1}j_{1}}\mathbb{E}_{\mathbb{P}}\left(A_{T,j_{1}}\right)=\sum_{j_{2}=1}^{n}\left.\partial_{\theta}\left(w(\theta)_{j_{2}}\right)\right|_{\theta=0}\mathbb{E}_{\mathbb{P}}^{c}\left(A_{T,i_{1}}A_{T,j_{2}}\right) ,
\label{eq:Beau1}
\end{equation}
for any $i_1 \in \{1,2,\ldots,n \}$, that are constraints connecting the mean and the covariance of the observable.

Analogous relations were found in \cite{Hur+14}, where they considered explicitly the case of rotations, inspired by the previous work \cite{AG07} on response functions. However, these kind of structural properties descend naturally from the FRs and therefore can be generalized to the present setting.

\subsection{Further Remarks}
\label{subsec:remarks}

In the next Section, we explicitly show that the previous result includes as particular cases the so-called `Spatial Fluctuation Relation'~\cite{PRG15} and `Symmetry-Induced Fluctuation Relation'~\cite{MPG20}.  
Before doing this, we can informally come back to the existence of an
asymptotic Fluctuation Relation for the scalar $\varSigma_{T}^{\mathbb{P}}/T$
as in~\eqref{eq:GC3}. By Varadhan contracting lemma~\cite{Var66,Tou09}, we obtain from \eqref{eq:ass1},\eqref{eq:Ldev},
the large deviation principle for $\varSigma_{T}^{\mathbb{P}}/T$ with the following rate function and cumulant generating function: 
\begin{equation}
\label{eq:marre}
  \begin{cases}
I_{\mathbb{P}}^{\Sigma^{\mathbb{P}}}\left(\sigma\right)=\underset{a : \left\langle w, a \right\rangle =\sigma}{\inf} I_{\mathbb{P}}^{A}\left(a\right)\\
\Lambda_{\mathbb{P}}^{\Sigma^{\mathbb{P}}}\left(k \right)= \underset{a}{\sup} \big( k\left\langle w,a\right\rangle -I_{\mathbb{P}}^A\left(a\right)\big)
\end{cases},  
\end{equation}
for $k,\sigma \in \mathbb{R}$.
The second equality also comes from the Legendre-Fenchel transform of $I_{\mathbb{P}}^{\Sigma^{\mathbb{P}}}$ after using the first equality, namely
\begin{align*}
 \Lambda_{\mathbb{P}}^{\Sigma^{\mathbb{P}}}\left(k \right) &=  \underset{\sigma}{\sup} \Big( k\sigma - \underset{a : \left\langle w, a \right\rangle =\sigma}{\inf} I_{\mathbb{P}}^A\left(a\right)\Big)  \\
 &= \underset{\sigma}{\sup} \Big( k\sigma + \underset{a : \left\langle w, a \right\rangle =\sigma}{\sup} \big(- I_{\mathbb{P}}^A\left(a\right)\big)\Big) \\
 &= \underset{\sigma}{\sup} \underset{a : \left\langle w, a \right\rangle =\sigma}{\sup} \Big( k\left\langle w, a \right\rangle - I_{\mathbb{P}}^A\left(a\right)\Big) ,
\end{align*}
so that in the end one optimizes over all possible $a \in \mathbb{R}^n$.
Then, the symmetry \eqref{eq:MR} implies that 
\begin{equation}\label{eq:symmetryScalar}
    \begin{cases}
I_{\mathbb{P}}^{\Sigma^{\mathbb{P}}}\left(\sigma\right)=\underset{a : \left\langle w, a \right\rangle =\sigma}{\inf} \big(\left\langle w,\mathcal{R}^{-1} a\right\rangle +I_{\mathbb{P}}^A\left(\mathcal{R}^{-1} a\right)\big)\\
\Lambda_{\mathbb{P}}^{\Sigma^{\mathbb{P}}}\left(k \right)=\underset{a}{\sup} \left( k \left\langle w, a \right\rangle -\left\langle w,\mathcal{R}^{-1} a \right\rangle -I_{\mathbb{P}}^A\left(\mathcal{R}^{-1} a \right)\right)
\end{cases}.
\end{equation}
But from this we cannot go forward and find a relation of type \eqref{eq:GC3} in general.

However, if the vector $w$ is a (right) eigenvector of\footnote{For instance, this is automatically the case if $\mathcal{R}$ is proportional
to the identity} the matrix $\mathcal{R}^{\dagger}$, i.e. $\mathcal{R}^{\dagger}w=\overline{\alpha}w$ for some $\alpha \in \mathbb{C}$ ($\overline{\alpha}$ is here the complex conjugate of $\alpha$),
then we obtain the following asymptotic Fluctuation Relations for the entropic functional 
\begin{equation}\label{eq:FRscalarAsymp}
   \begin{cases}
I_{\mathbb{P}}^{\varSigma^{\mathbb{P}}}\left(\sigma\right)=\frac{\sigma}{\alpha}+I_{\mathbb{P}}^{\varSigma^{\mathbb{P}}}\left(\frac{\sigma}{\alpha}\right)\\
\Lambda_{\mathbb{P}}^{\varSigma^{\mathbb{P}}}\left(k\right)=\Lambda_{\mathbb{P}}^{\varSigma^{\mathbb{P}}}\left(k\alpha-1\right)
   \end{cases}, 
\end{equation}
which generalize~\eqref{eq:GC3} to the case $\alpha\neq-1$ (in the usual involution case $\alpha=-1$). Indeed, one can write from the first of \eqref{eq:symmetryScalar}
\begin{align*}
   I_{\mathbb{P}}^{\Sigma^{\mathbb{P}}}\left(\sigma\right) &=\underset{a : \left\langle w, \mathcal{R} a \right\rangle =\sigma}{\inf} \big(\left\langle w, a\right\rangle +I_{\mathbb{P}}^A\left( a\right)\big) \\
   &= \underset{a : \alpha\left\langle w, a \right\rangle =\sigma}{\inf} \big(\left\langle w, a\right\rangle +I_{\mathbb{P}}^A\left( a\right)\big) \\
   &= \frac{\sigma}{\alpha} +  \underset{a : \left\langle w, a \right\rangle =\frac{\sigma}{\alpha}}{\inf} I_{\mathbb{P}}^A\left( a\right) ,
\end{align*}
that gives the first of \eqref{eq:FRscalarAsymp}. The second equation also follows easily from the second of \eqref{eq:symmetryScalar}
\begin{align*}
    \Lambda_{\mathbb{P}}^{\Sigma^{\mathbb{P}}}\left(k \right) &=\underset{a}{\sup} \left( k \left\langle w, \mathcal{R} a \right\rangle -\left\langle w, a \right\rangle -I_{\mathbb{P}}^A\left( a \right)\right) \\
    &= \underset{a}{\sup} \left( (k \alpha - 1) \left\langle w,  a \right\rangle  -I_{\mathbb{P}}^A\left( a \right)\right) .
\end{align*}

The following table provides a summary of our fluctuation relations for vector-valued observables and general bijections, and visually compares them to those obeyed by the scalar entropic functional under involutions. Note that our framework can also give new relations for vector-valued observables in the case of involutions (in this case one has $\mathcal{R}^{-1}= \mathcal{R}$ in the right column).

\begin{table}[!h] 

\begingroup
\setlength{\tabcolsep}{4pt}
\renewcommand{\arraystretch}{1.8}
\begin{center}
\begin{tabular}{ |c|c|c| } 
 \hline
 Involution (scalar obs.) &  General Bijection (vector obs.)\\
 \hline
 $\frac{d\mathcal{P}_{\varSigma_{T}^{\mathbb{P}}}}{d\mathcal{P}_{-\varSigma_{T}^{\mathbb{P}}}} (\sigma) = \exp(\sigma) $  & $\frac{d\mathcal{P}_{A_T}}{d\mathcal{P}_{\mathcal{R}A_T}}(a)=\exp \left(-\left\langle w,\mathcal{R}^{-1} a\right\rangle \right)$ \\ $\mathbb{E}_{\mathbb{P}}\Big[\exp\left(k\varSigma_{T}^{\mathbb{P}}\right)\Big]=\mathbb{E}_{\mathbb{P}}\Big[\exp\left(-\left(k+1\right)\varSigma_{T}^{\mathbb{P}}\right)\Big]$ & $\mathbb{E}_{\mathbb{P}}\Big[\exp\Big\langle k,A_T\Big\rangle \Big]=\mathbb{E}_{\mathbb{P}}\Big[\exp \Big\langle \mathcal{R}^{\dagger} k - w,A_T\Big\rangle \Big] $ \\
 $I_{\mathbb{P}}\left(\sigma\right)=-\sigma+I_{\mathbb{P}}\left(-\sigma\right)$ 
& $  I_{\mathbb{P}}\left(a \right)=\left\langle w,\mathcal{R}^{-1} a \right\rangle +I_{\mathbb{P}}\left(\mathcal{R}^{-1} a \right) $   \\ 
 $\Lambda_{\mathbb{P}}(k)=\Lambda_{\mathbb{P}}(-k-1) $  &  $\Lambda_{\mathbb{P}}\left( k \right)=\Lambda_{\mathbb{P}}\left(\mathcal{R}^{\dagger} k -w\right)$ \\ 
 
 \hline
\end{tabular}

\end{center}
\endgroup

\end{table}

A further remark is required at this point about the assumptions and the mathematical structure we consider.  The aim of this paper is to provide a general framework to unify a class of physically important fluctuation relations associated with transformations other than time-reversal. In order to show that the symmetry properties can be proved rigorously under suitable assumptions, for simplicity we choose to make hypotheses that are strong (and probably unnecessarily strong, especially in the large deviation setting). This is because we do not want to obscure the main message of the paper with technical difficulties related to weakening the assumptions. Even in the case of time-reversal, challenges arise when trying to rigorously establish the existence of a large deviation principle for stochastic observables, mainly due to non-compact state space and unbounded components of the entropic functional (see e.g. \cite{BdG15,JPS17,Raq24}). Also, the choice to pass through the G\"artner-Ellis theorem may not be the unique way or the most suitable one to attack the problem, and different approaches, like contraction principles from higher-level (e.g. level 3 or level 2.5) large deviation functionals, may be convenient (see \cite{BGL23} for something in this direction).

The next two sections are dedicated to the analysis of the presented general framework in specific situations. In particular, we focus on two case studies that provide an explicit connection between the seemingly abstract assumptions stated above and the symmetry properties of the underlying processes. This also allows a better comparison with the existing literature \cite{Hur+11,Hur+14,LG14,PRG15,MPG20}. Note, however, that the analyzed models are not exhaustive and one could apply the general results to further situations. For instance, stochastic processes describing quantum systems, like Markovian quantum-jump unravelings \cite{WM10}, would fit into the scheme. In fact, one could think of using the framework developed here to generalize the fluctuation relations derived in \cite{MPG21} for Markovian quantum-jump processes to non-Markovian unravelings \cite{SCP20}.

\section{Case study 1: FRs for Canonical Path Probabilities}
\label{sec:canonical}

\subsection{Results for Generic Canonical Processes}

In this Section, we consider a physical process $X_{t}$ such that its path probability is the Canonical path probability \cite{Gar+09, JS10, CT13, CT15a, CT15b, JS15} associated to an a priori path measure $\mathbb{P}_{\left[0,T\right]}^{0}$
(not necessarily Markovian) and a biasing vector-valued observable $K_T: \Omega_T \to \mathbb{R}^{n}$: 
\begin{equation}\label{eq:cano}
   \mathbb{P}_{\left[0,T\right]}^{\mathrm{cano},s} \equiv\frac{\mathbb{P}_{\left[0,T\right]}^{0}\,  \mathrm{e}^{\left\langle s,K_T \right\rangle}  }{\mathbb{E}_{\mathbb{P}^{0}} \!\left[\mathrm{e}^{\left\langle s,K_T\right\rangle }\right]}, 
\end{equation}
with the biasing vector $s$. This is well-defined for all values of $s \in S\subseteq \mathbb{R}^n$ such that $0< \mathbb{E}_{\mathbb{P}^{0}} \!\left[\mathrm{e}^{\left\langle s,K_T\right\rangle }\right] < \infty$. Such a process has been called canonical in \cite{CT15a} because it is a generalization to nonequilibrium settings of the canonical ensemble of equilibrium statistical mechanics but it appears in the literature also with other names, like tilted ensemble, or $s$-ensemble (due to the parameter usually labeled $s$). In the mathematics literature these processes are also called \emph{penalizations} \cite{RY09}.

Applying our general formalism presented in the last section to this
path probability, the scalar entropic functional \eqref{eq:AF} can be readily obtained and explicitly reads
$$
\varSigma_{T}^{\mathbb{P}^{\mathrm{cano},s}}=\varSigma_{T}^{\mathbb{P}^{0}}-\left\langle s, \, K_T \circ R - K_T \right\rangle .
$$
In order to proceed we resort to the following Assumption.
\begin{assumption}[Symmetry and observable covariance] \label{ass_cano}
The process satisfies these two conditions:\\
$\bullet$ The a-priori process is $R$-symmetric, 
\begin{equation}\label{eq:h1cano} d\mathbb{P}_{\left[0,T\right]}^{0}\left[R\left(X_{\left[0,T\right]}\right)\right]=d\mathbb{P}_{\left[0,T\right]}^{0}\left[X_{\left[0,T\right]}\right],
\end{equation}
so that one has $\varSigma_{T}^{\mathbb{P}^{0}}=0$. \\
$\bullet$  There exists a space-time homogeneous invertible $n \times n$ matrix $\mathcal{R}$ such that, for any $X_{\left[0,T\right]}$,
\begin{equation}\label{eq:hypothese-1-1}
K_T \left(R\left(X_{\left[0,T\right]}\right)\right)=\mathcal{R} K_T\left(X_{\left[0,T\right]}\right) .
\end{equation}
\end{assumption}

These two hypotheses imply the Assumption \ref{ass} with the vector $w$ expressed in terms of the biasing field $s$ as follows (here $\bbbone_n$ is the $n\times n$ identity matrix), 
\begin{equation}
 \label{eq:w-ici}   
w =\left(\bbbone_n-\mathcal{R}^{\dagger}\right) s, 
\end{equation}
and the role of $A_T$ is played by the biasing observable $K_T$. From now on we use this notation for the biasing observable $K_T \equiv A_T$. Therefore, the following result is a consequence of Theorem \ref{thm:result1} when $S=\mathbb{R}^n$.
\begin{corollary}[Finite-time FRs for Canonical Processes] 
Under Assumption \ref{ass_cano}, for all times $T \geq 0$ and for all $a,k,s \in \mathbb{R}^n$, the canonical process \eqref{eq:cano} satisfies the FRs 
    \begin{align}    
\frac{d \mathcal{P}_{A_T}^{\mathrm{cano,s}}}{d \mathcal{P}_{\mathcal{R}A_T}^{\mathrm{cano,s}} }(a) &=\exp \left\langle (\mathcal{R}^\dag - \bbbone_n) s,\mathcal{R}^{-1} a\right\rangle \; , \label{eq:canoVFR1} \\
\mathbb{E}_{\mathbb{P}^{\mathrm{cano},s}}\Big[\exp\Big\langle k,A_T\Big\rangle \Big] &=\mathbb{E}_{\mathbb{P}^{\mathrm{cano},s}}\Big[\exp \Big\langle \mathcal{R}^{\dagger} (k + s) - s, A_T\Big\rangle \Big] , \label{eq:canoVFR2}
    \end{align}
     where $\mathcal{P}_{A_T}^{\mathrm{cano,s}} \equiv \mathbb{P}^{\mathrm{cano,s}} \circ A_T^{-1}$.
\end{corollary}

Concerning the validity of Assumption \ref{ass:ld}, due to the particular construction of the canonical path measure, one has
\begin{equation}
    \mathbb{E}_{\mathbb{P}^{\mathrm{cano},s}} \left[ \mathrm{e}^{\left\langle k, A_T \right\rangle} \right] = \frac{\mathbb{E}_{\mathbb{P}^0} \left[ \mathrm{e}^{\left\langle k+s, A_T \right\rangle} \right]}{\mathbb{E}_{\mathbb{P}^0} \left[ \mathrm{e}^{\left\langle s, A_T \right\rangle} \right]} ,
\end{equation}
therefore it is sufficient to assume that the quantity
\begin{equation}\label{eq:cumulp0}
    \Lambda_{0}(k) = \lim_{T \to \infty} \frac{1}{T} \ln \mathbb{E}_{\mathbb{P}^0} \left[ \mathrm{e}^{\left\langle k, A_T \right\rangle} \right]
\end{equation}
exists and is finite for the a priori probability. In fact, if $\Lambda_0 (k)$ is finite for any $k$, then $\Lambda_{\mathbb{P}^{cano,s}}(k)$ is also finite and it explicitly reads 
\begin{equation}
\label{eq:jolie}    
\Lambda_{\mathbb{P}^{cano,s}}(k) = \Lambda_0 (k+s) - \Lambda_0 (s).
\end{equation}
Likewise, for the rate function one has
\begin{align}
    I_s(a) = \sup_{k \in \mathbb{R}^d}\Big( \langle k, a \rangle - \Lambda_s(k) \Big) &= \sup_{k \in \mathbb{R}^d} \Big( \langle k +s , a \rangle - \Lambda_0(k + s) \Big) - \langle s, a \rangle + \Lambda_0(s)  \nonumber\\
    &= I_0(a) - \langle s, a \rangle + \Lambda_0(s) ,
    \label{eq:jolieRF} 
\end{align}   
where the lighter notation $\Lambda_s \equiv \Lambda_{\mathbb{P}^{cano,s}}$ and $I_s \equiv I_{\mathbb{P}^{\mathrm{cano},s}} $ has been introduced for convenience.
Applying then the main result Theorem \ref{thm:result2} of the previous section, we obtain the following Corollary.
\begin{corollary}[Asymptotic FRs for Canonical Processes]  
The canonical process \eqref{eq:cano}, under Assumption \ref{ass_cano} and provided \eqref{eq:cumulp0} is finite and differentiable on $\mathbb{R}^n$, satisfies the asymptotic FRs for the observable $A_T/T$
\begin{equation}\label{eq:MR-1}
     \begin{cases}
 I_s\left(a\right)=\left\langle \left(\bbbone_n-\mathcal{R}^{\dagger}\right) s, \, \mathcal{R}^{-1} a\right\rangle +I_s\left(\mathcal{R}^{-1} a\right)\\
\Lambda_s \left(k\right) = \Lambda_s \left(\mathcal{R}^{\dagger}k -\left(\bbbone_n -\mathcal{R}^{\dagger}\right) s\right) 
     \end{cases}.
\end{equation}
\end{corollary}
Notice that the latter equations can also be written as 
\begin{equation}\label{eq:MR-1-1}
\begin{cases}
I_s \left(\mathcal{R} a\right)=\left\langle s ,\left(\bbbone_n -\mathcal{R}\right) a \right\rangle +I_s\left(a\right)\\
\Lambda_s\left(k\right) = \Lambda_s \left(\mathcal{R}^{\dagger}\left(k +s\right) -s\right) 
\end{cases}.
\end{equation}
These symmetries in a similar set-up are the main result of \cite{PRG15} (relation 5-6) where the name \emph{spatial fluctuation relations} was used and of \cite{MPG20} (relation 12) where instead they were called \emph{symmetry-induced fluctuation relations}, to highlight the fact that $R$ is not necessarily a spatial transformation. However, the mentioned works did not highlight the generality of the processes one can consider, for instance removing the Markovianity constraint. We will comment more on this with our first example. On the technical side, when the cumulant generating function is finite on a proper subset $S \subset \mathbb{R}^n$, as in the second example of this Section, things are a bit more complicated, even though one could still use G\"artner-Ellis for $\Lambda$ essentially smooth. More importantly, when the steepness condition at the boundary fails there is no general result available and large deviation principles have to be proved with ad-hoc techniques, even in the usual case of scalar observables \cite{JPS17}.

As a final remark, we note that the asymptotic result could reasonably apply to processes whose path measure is not exactly \eqref{eq:cano} at finite time but becomes equivalent to it for long times (see for instance the relation between the canonical process and the so-called driven process in \cite{CT15a}).

In the same vein, Assumption \ref{ass_cano} is sufficient but not necessary to obtain the FRs \eqref{eq:MR-1-1} via the relations \eqref{eq:jolie} and \eqref{eq:jolieRF}. Indeed, we see that the symmetry conditions 
\[
\begin{cases}\label{eq:sym}
\Lambda_{0}\left(\mathcal{R}^\dag s\right)=\Lambda_{0}\left(s\right)\\
I_{0}\left(\mathcal{R}a\right)=I_{0}\left(s\right)
\end{cases}, 
\]
 are sufficient to obtain \eqref{eq:jolie}, and are weaker than  Assumption \ref{ass_cano} (i.e. Assumption \ref{ass_cano} implies them but not viceversa).

Finally, in this Canonical Path probability set-up, let us consider a one-parameter family of transformations $\mathcal{R}(\theta)$, all of them satisfying Assumption \ref{ass_cano} with the same observable (independent of $\theta$) and with $w(\theta)$ given as in \eqref{eq:w-ici} for each $\theta$. Recalling the definition of the generator $\mathcal{N}$ through the matrix expansion $\mathcal{R}(\theta)=\bbbone_n +\theta\mathcal{N}+o\left(\theta\right)$, we obtain $\partial_{\theta}w(\theta)|_{\theta=0}=-\mathcal{N}^{\dagger}s$,
and the hierarchy \eqref{eq:BeauG} becomes 
\begin{align}
\label{eq:SuperG}
\sum_{k=1}^{p}\sum_{j_{k}=1}^{n}\mathcal{N}_{i_{k}j_{k}}\mathbb{E}_{\mathbb{P}^{cano,s}}^{c}&\left(A_{T,i_{1}}A_{T,i_{2}}...A_{T,i_{k-1}}A_{T,j_{k}}A_{T,i_{k+1}}...A_{T,i_{p}}\right) \nonumber\\
&= -\sum_{j_{p+1}=1}^{n}\sum_{j_{p+2}=1}^{n}s_{j_{p+2}} \mathcal{N}_{j_{p+2}j_{p+1}}\mathbb{E}_{\mathbb{P}^{cano,s}}^{c}\left(A_{T,i_{1}}A_{T,i_{2}}...A_{T,i_{p}}A_{T,j_{p+1}}\right) ,
\end{align}
for any integer $p$ and any set of indices $(i_1,\ldots, i_p) \in \{1,2,\ldots,n \}^p$.
In particular, the case $p=1$ gives the following relations between the mean and the covariance
\begin{equation}
\label{eq:Super1}
\sum_{j_{1}=1}^{n}\mathcal{N}_{i_{1}j_{1}}\mathbb{E}_{\mathbb{P}^{cano,s}}\left(A_{T,j_{1}}\right)=-\sum_{j_{2}=1}^{n}\sum_{j_{3}=1}^{n}s_{j_{3}}\mathcal{N}_{j_{3}j_{2}}\mathbb{E}_{\mathbb{P}^{cano,s}}^{c}\left(A_{T,i_{1}}A_{T,j_{2}}\right),
\end{equation}
for any $i_1 \in \{1,2,\ldots,n \}$. These relations generalize to arbitrary observables and transformations the hierarchy presented in \cite{Hur+14} (equation (64)) for currents and rotations. Further investigation on their physical relevance is certainly worth and it is ongoing.

\subsection{Example: Canonical processes associated to time-homogeneous Semi-Markov processes}
\label{sec:ex1}

\textbf{Elements of time-homogeneous Semi-Markov processes}

\noindent
To highlight the fact that Markovianity of the process is never used in deriving the results of the previous section, we consider here the explicit example of time-homogeneous Semi-Markov jump processes with finite state space $E$. For a general account of these processes and their relevance in physics applications we refer to \cite{AG08,MNW09}. As in continuous-time Markov chains a trajectory is specified by the sequence of visited configurations $x_i \in E$ and jump times $t_i$. However, in this case the sojourn times are not distributed exponentially thus making the process non-Markovian. The transitions from configuration $x$ to configuration $y$ \emph{after} a sojourn time $t$ are characterized by a semi-Markov kernel $M(x,y,t)$, with $\sum_y\int_0^\infty M(x,y,t)\, dt=1$. The transition matrix $p(x,y)$ from $x$ to $y$ and the \emph{sojourn time} distribution in $M(x,t)$ are then obtained as follows
\begin{equation*}
    p(x,y)= \int_0^\infty M(x,y,t)\, dt ,  \qquad M(x,t) = \sum_y M(x,y,t).
\end{equation*}
For simplicity, in this manuscript we also assume the property of time-direction independence, namely $M(x,y,t)=p(x,y)M(x,t)$. This was shown to be a necessary (but not sufficient) condition for the reversibility of $\mathbb{P}_{[0,T]}^0$ \cite{QW06}.
The probability density of a trajectory $X_{[0,T]}:= (x_0,x_1,...,x_n;t_1,...,t_n,T)$ can be written explicitly
\begin{align}\label{eq:ptraj}
    d \mathbb{P}_{[0,T]}^0 [X_{[0,T]}] = \pi(x_0) \prod_{i=0}^{n-1} M(x_i,x_{i+1},t_{i+1}-t_i) \left(\int_{T-t_n}^{\infty} M(x_n,\tau) d\tau \right) \prod_{j=1}^{n} dt_j \;,
\end{align}
where $\pi(x_0)$ is the probability of the initial condition $x_0$ and we consider $t_0=0$. In writing the previous formula, we assumed that $t_1-t_0$ corresponds to the sojourn time in $x_0$, or equivalently that the system has jumped to the configuration $x_0$ at time $t=0$. Therefore, we do not average over the negative-time histories as done in \cite{MNW09} to take care of the unknown sojourn time in $x_0$. \\
\textbf{Discrete translation on a ring for a semi-Markov Random walk} \\
\noindent
In order to check if the condition \eqref{eq:h1cano} is satisfied, one needs to specify a transformation $R$. For concreteness, we assume the state space to be $E= \{ 1,2, \ldots, L \}$ and take $R$ to be a coordinate transformation $R X_{[0,T]}= (r x_0,r x_1,...,r x_n;t_1,...,t_n,T)$ with
$$r x = x+1 \; \mathrm{MOD} \; L \;,$$
namely we shift all coordinates by $1$ with periodic boundary conditions.
For example, a process that is symmetric under this transformation is the one specified by the following semi-Markov kernel
\begin{equation}
M(x,y,t)=\gamma_{i}(t)\quad  \mathrm{if} \quad y=x+i \; \mathrm{MOD} \; L,
\label{eq:gen}
\end{equation}
with $i \in \{1,2,\ldots,L-1\}$, describing a random walker on a ring possibly jumping on any other site with a rate that depends on the clockwise distance between the departure and arrival sites. This system is translation invariant, with stationary probability given by the uniform measure on the state space. Therefore, the fluctuation relations \eqref{eq:canoVFR1}-\eqref{eq:canoVFR2} are satisfied for any biasing observable that is covariant with respect to $R$, namely, that obeys \eqref{eq:hypothese-1-1}. In the following, we consider the special case of nearest neighbor jumps, that is
 \begin{equation}\label{eq:gen2}
    \gamma_{i}(t)=\gamma_{+}(t)\delta_{i,1}+\gamma_{-}(t)\delta_{i,L-1}. 
 \end{equation}
In this way, we can easily associate an order relation to the set of bonds (site pairs where the jump is allowed) and study counting observables related to jumps in particular subsets.
\\
\textbf{Bidimensional Observable} \\
\noindent
For instance, one can take $A_T$ to be the two-dimensional observable whose components are the number of jumps on even and odd bonds\footnote{We consider the bond $1 \leftrightarrow 2$ to be the first one and it is therefore labeled as odd. We then label the bond $2 \leftrightarrow 3$ as even, and so on and so forth.} (provided the number of sites is even so that an even bond is always followed by an odd bond and there is an equal number of even and odd bonds).
In this case, one has \footnote{Note that $\mathcal{R}$ is an involution,  $\mathcal{R}^2=\bbbone_2$ , despite the fact that the shift $R$ is not an involution.}
\begin{equation}
\varSigma_{T}^{\mathbb{P}^{\mathrm{cano},s}}= \left\langle s, \,  A_T  - \mathcal{R} A_T \right\rangle, \quad\quad \mathcal{R}= 
\begin{pmatrix}
    0 & 1 \\
    1 & 0
\end{pmatrix}.
\end{equation}
The finite-time FRs \eqref{eq:canoVFR1}-\eqref{eq:canoVFR2} therefore explicitly read
 \begin{align}
\frac{d \mathcal{P}^{\mathrm{cano,s}}_{A_T}}{d  \mathcal{P}^{\mathrm{cano,s}}_{\mathcal{R}A_T}}(a_{1},a_{2}) &= \exp \Big( (s_{1}-s_{2})(a_{1}-a_{2}) \Big) \; , \label{eq:canoVFR1-2} \\
\mathbb{E}_{\mathbb{P}^{\mathrm{cano},s}} \Big[\exp\Big(
k_{1} A_{T}^{(1)}+ k_{2} A_{T}^{(2)} \Big) \Big] &=\mathbb{E}_{\mathbb{P}^{\mathrm{cano},s}}\Big[ \exp \Big( (k_{2}+s_{2}-s_{1}) A_{T}^{(1)} + (k_{1}+s_{1}-s_{2}) A_{T}^{(2)} \Big) \Big] , \label{eq:canoVFR2-1}
    \end{align} 
for $k=(k_1,k_2) \in \mathbb{R}^2$ and $a=(a_1,a_2) \in \mathbb{N}^2$.
The relation \eqref{eq:canoVFR1-2} gives a precise quantification of the unbalance between the activity (number of jumps) on even and odd bonds in the canonical process. This unbalance is caused by the biasing vector $s$, that can be interpreted as a physical field breaking the original symmetry of the process $\mathbb{P}^0$. Nevertheless, a reminiscence of the broken symmetry is still present on the statistics of the fluctuating observable $A_T$, as described by the relation on the moment generating function \eqref{eq:canoVFR2-1}.    \\
\textbf{Explicit Large Deviations} \\
\noindent
Moreover, the moment generating function and the cumulant generating function $\Lambda_s\left(k\right)$  can computed exactly for a typical choice of the semi-Markov kernel \eqref{eq:gen}-\eqref{eq:gen2} that is
\begin{equation*}
  \begin{cases}
\gamma_{+}(t)=p\frac{\lambda^{n}t^{n-1}}{\left(n-1\right)!}\exp\left(-\lambda t\right)\\
\gamma_{-}(t)=q\frac{\lambda^{n}t^{n-1}}{\left(n-1\right)!}\exp\left(-\lambda t\right),
\end{cases}  
\end{equation*}
with $p+q=1$,  $n \in \mathbb{N}$,  and $n\geq 1$. The value $n=1$ corresponds to the Markov case. In order not to overshadow the main topic of the paper, we present here the formula for the scaled cumulant generating function and postpone the detailed calculations in Appendix \ref{app_semimarkov}. Explicitly, one has for the symmetric process $\mathbb{P}^0$
\begin{equation}\label{eq:lambdaSemiMarkov}
    \Lambda_0(k)= \lambda \left( (p \mathrm{e}^{k_1} + q \mathrm{e}^{k_2})(q \mathrm{e}^{k_1} + p \mathrm{e}^{k_2}) \right)^{1/2n} - \lambda.
\end{equation}
Some level curves of $\Lambda_0(k)$ are plotted in Figure \ref{Fig:ld} panel (i) to highlight the symmetry under the exchange $k_1 \leftrightarrow k_2$.
This result can be compared with the result in \cite{MNW09} where the tilting is done in terms of the current observable, so that there $k_1=-k_2=k$ and $p$ is always paired with $\mathrm{e}^{k_1}$ while $q$ is always paired with $\mathrm{e}^{-k_1}$ because jumps to $x+1$ always give positive current and jumps to $x-1$ always give negative current, independently whether they are associated with even or odd bonds. The function $\Lambda_0(k)$ is well-defined for any $k \in \mathbb{R}^2$ and it is everywhere differentiable. Therefore, a large deviation principle holds true for the family of probability measures $\mathcal{P}_{A_T/T}^0(a)$, $a = (a_1,a_2) \in \mathbb{R}_+^2$, and the rate function can be computed explicitly (see Appendix \ref{app_semimarkov}),
\begin{align}\label{eq:ratefunctionSemiMarkov}
    I_0(a) &= n(a_1 + a_2)  \ln \left(\frac{n (a_1+a_2)}{\lambda}\right)  + a_1 \ln\left( \sqrt{1 + f(a)} + \frac{\mathrm{sign}(a_1-a_2)}{|p-q|}\sqrt{ f(a)-1} \right) \nonumber\\
    &+a_2 \ln\left( \sqrt{1 + f(a)} - \frac{\mathrm{sign}(a_1-a_2)}{|p-q|}\sqrt{ f(a)-1} \right) - n (a_1 + a_2) + \lambda - \frac{(a_1+a_2)}{2}\ln(2)
\end{align}
where we defined the function $f(a)$ as follows to enhance the readability
\begin{equation}
    f(a)= \sqrt{1 + \frac{(p-q)^2}{4 p^2 q^2}\left(\frac{a_1-a_2}{a_1+a_2} \right)^2 }.
\end{equation}
It is immediate to see that $I_0(a)$ (resp $\Lambda_0(k)$) are  indeed symmetric under the exchange $a_1 \leftrightarrow a_2$ (resp $k_1 \leftrightarrow k_2$):  they fulfill the FRs \eqref{eq:MR-1-1} in the case $s=0$, for any \emph{non-Markovianity degree} $n$. In Figure \ref{Fig:ld} panel (ii) some level curves of $I(a)$ are plotted in the case $n=2$. The expression \eqref{eq:ratefunctionSemiMarkov} represents one of the very few examples of rate function that can be computed explicitly for a non-Markovian dynamics, and therefore it has an interest \emph{per se}, beyond the study of Fluctuation Relations. We analyze it more in detail in the following. In particular, we focus on the dependence on $n$ for fixed $a$, to see whether non-Markovianity enhances or suppresses fluctuations. We promote $n$ to be a positive real variable $x$ for the moment and study the function 
\begin{equation}
    I_0(a,x) = x (a_1+a_2) \ln \left(\frac{x (a_1+a_2)}{\lambda}\right) - x (a_1 + a_2) + g(a),
\end{equation}
where $g(a)$ is any function of $a$ independent of $x$.
The derivative with respect to $x$ reads
\begin{equation}
    \partial_x I_0(a,x) = (a_1+a_2) \ln \left(\frac{x (a_1+a_2)}{\lambda}\right) ,
\end{equation}
so that it is negative for $x< x^* := \lambda/(a_1+a_2)$ and non-negative for $x \geq x^*$. Recalling that we are interested in integer $n \geq 1$, we see that if $x^* <1$, i.e. for $a_1+a_2$ large enough, the rate function is monotonically increasing in $n$. Conversely, for small values of $a_1+a_2$ the rate function is not monotonic: it decreases up to a certain value $\overline{n}$ and it starts increasing for $n> \overline{n}$. Therefore, at least in this case, non-Markovianity suppresses fluctuations towards large values of the observable, while fluctuations towards small values can be enhanced or suppressed with respect to the Markovian case, depending on $n$. This situation is represented in Figure \ref{Fig:ld} panel (iii) for the special case $a_1=a_2$. It is certainly worth studying further how generic this behavior is, also in light of recent progress on the subject. Indeed, large deviation principles for Semi-Markov processes have been investigated recently, both in the physics literature \cite{AG08,MNW09,CH16,SK18,VH20} and in the probability literature \cite{MZ16,Fag17,JJW24}. In particular, a large deviation principle for the joint law of the empirical measure and the empirical flow (level $2.5$) is rigorously proved in \cite{MZ16} in the case of finite configuration space, and in \cite{JJW24} for countable configuration space.

\begin{figure}[ht]
\centering
\vspace*{-0.4cm}
\includegraphics[width=.96\textwidth]{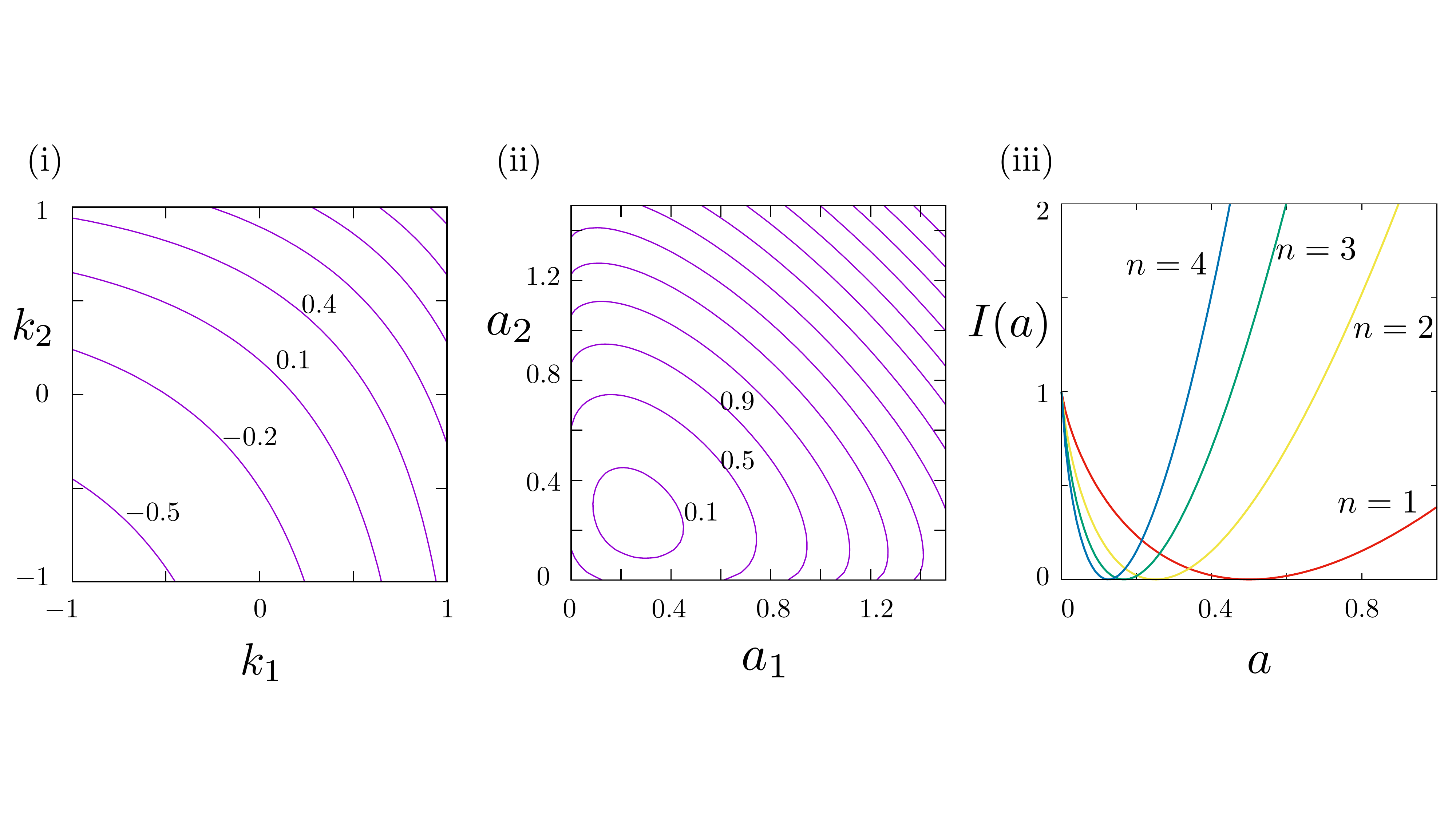}
\vspace*{-1.5cm}
\caption{(i) Level curves of the function $\Lambda_0(k)$, \eqref{eq:lambdaSemiMarkov}, for $p=0.7, q=0.3,\lambda=1$ and $n=2$. The curve on the left bottom corresponds to $\Lambda_0=-0.5$ and the subsequent ones correspond to increments of $0.3$. (ii) Level curves of the function $I_0(a)$, \eqref{eq:ratefunctionSemiMarkov}, for $p=0.7, q=0.3,\lambda=1$ and $n=2$. The closed curve on the bottom left corresponds to $I_0=0.1$ and the other curves are related to increments of $0.4$. (iii) Rate function $I_0$ for $a_1=a_2=a$ and different values of $n$. As before we set $p=0.7,q=0.3,\lambda=1$. One can see the non-monotonic behavior with respect to $n$ for small $a$, like for instance around $a=0.2$.}
\label{Fig:ld} 
\end{figure}

\noindent
\textbf{Fluctuation Relations} \\
\noindent 
Coming back to fluctuation relations, the large deviation properties of the canonical process can be readily inferred from the ones of the symmetric process via the relations \eqref{eq:jolie} and \eqref{eq:jolieRF}.
Indeed, the scaled cumulant generating function for the canonical process reads
\begin{equation}\label{eq:lambdacano}
    \Lambda_s(k) = \lambda \Big( (p \mathrm{e}^{s_1 + k_1} + q \mathrm{e}^{s_2 + k_2})(q \mathrm{e}^{s_1 + k_1} + p \mathrm{e}^{s_2 +k_2}) \Big)^{1/2n} -\lambda \Big( (p \mathrm{e}^{s_1} + q \mathrm{e}^{s_2})(q \mathrm{e}^{s_1} + p \mathrm{e}^{s_2}) \Big)^{1/2n},
\end{equation}
and the rate function 
\begin{equation}\label{eq:Icano}
I_s(a) = I_0(a) - \langle s, a \rangle + \Lambda_0(s).  
\end{equation}    
Therefore, it is easy to check that the asymptotic FRs \eqref{eq:MR-1-1}, taking the form 
\begin{equation}\label{eq:MR-1-1-1}
\begin{cases}
I_s \left(a_{2},a_{1}\right)= \left(s_{1}-s_{2}\right)\left(a_{1}-a_{2}\right) +I_s\left(a_{1},a_{2}\right) \\
\Lambda_s \left( k_{1},k_{2} \right) = \Lambda_s \left( k_{2}+s_{2}-s_{1},k_{1}+s_{1}-s_{2}\right) 
\end{cases},
\end{equation}
hold true for this example. Indeed, concerning the rate function, one can use the symmetry of $I_0$ and  $\Lambda_0$ and write  with \eqref{eq:Icano} 
\begin{equation*}
   I_s \left(a_{2},a_{1}\right) - I_s \left(a_{1},a_{2}\right) = -(s_1 a_2 + s_2 a_1) + (s_1 a_1 + s_2 a_2) = \left(s_{1}-s_{2}\right)\left(a_{1}-a_{2}\right).
\end{equation*}
Also, for the $\Lambda_s(k)$ the symmetry can be checked noticing that when $k_1$ is substituted with $ k_2+s_2-s_1$, the sum $k_1+s_1$ is mapped into $k_2+s_2$, and similarly for $k_2$. Therefore, one has only exchanged the two factors in the first term of \eqref{eq:lambdacano}.

We conclude this example with a remark.
Note that since $\bbbone_2 - \mathcal{R}$ is not invertible, the entropic functional does not really depend on the two components of $A_T$ independently but it reads $\varSigma_{T}^{\mathbb{P}^{\mathrm{cano},s}} = (A_T^{(1)} - A_T^{(2)})(s_1 - s_2)$. Therefore, in this case it is possible to find a fluctuation relation for the scalar contracted observable $ \widetilde{A}_T := A_T^{(1)} - A_T^{(2)}$ (recall the discussion at the end of Section \ref{sec:General--transformation}).
In order to have an entropic functional with a truly vectorial dependence one can enlarge the dimension of the observable and consider for instance a tripartite structure for the bonds (bonds of type $3k, 3k+1, 3k+2$). The matrix $\mathcal{R}$ is the following permutation matrix
$$
\mathcal{R}=\left(\begin{array}{ccc}
0 & 1 & 0\\
0 & 0 & 1\\
1 & 0 & 0
\end{array}\right)
$$
and the functional will depend separately on $A_T^{(1)}-A_T^{(2)}$ and $A_T^{(2)}-A_T^{(3)}$. More in general, one can think of a $N$-fold structure for the bonds and obtain $N-1$ independent components. With this choice of observable, the FRs \eqref{eq:canoVFR1}-\eqref{eq:canoVFR2} then quantify the imbalance between the number of jumps in different kinds of bonds in the canonical process (the original process is symmetric as one can easily recall by setting $s=0$).  Also, we can mention that in the spirit of the weaker decomposition in Assumption \ref{ass_weak}, one could consider an  observable containing a subleading contribution that does not affect the large deviation statistics. For instance the scalar observable $\widetilde{A}_T := A_T^{(1)} - A_T^{(2)}$ in the large time limit would have the same scaled cumulant generating function, and therefore the same symmetry, of the observable $\widetilde{A}_T+ \widetilde{B}_T$ where $\widetilde{B}_T$ is the difference between the number of jumps on even sites and odd sites. Indeed, while the difference between the jumps on even and odd bonds can be arbitrarily large, a jump from an even site is always followed by a jump from an odd site, and therefore the difference is at most one for any time $T$.

\subsection{Example: Langevin equation with harmonic potential and rotations in momentum space}

The discussion of this section is not confined to discrete systems. We study here a second example where the phase space is $\mathbb{R}^4$. 
Consider a particle with unit mass moving in $\mathbb{R}^2$ and experiencing a Langevin dynamics with harmonic potential $V(Q_1,Q_2)= (Q_1^2 + Q_2^2)/2$. The system is described by the following SDE 
\begin{equation}\label{eq:langevin}
    dQ_{i,t} = P_{i,t} dt , \quad dP_{i,t}= - Q_{i,t} dt - \gamma P_{i,t} dt + \sqrt{2\gamma \beta^{-1}} dW_{i,t} , \quad i \in\{1,2\},
\end{equation}
where $Q_{i,t}, P_{i,t}$ are respectively the position and momentum variables at time $t$, and $W_{i,t}$ are standard Wiener processes. The parameters $\gamma, \beta \in \mathbb{R}_+$ represent the damping coefficient and the inverse temperature. The corresponding formal Markov generator is the sum of two independent contributions $L=L^{(1)}+L^{(2)}$, each of them reading
\begin{equation}
    L^{(i)}= p_i \, \partial_{q_i} - q_i \,\partial_{p_i} -\gamma p_i \,\partial_{p_i} + \frac{\gamma}{\beta} \partial^2_{p_i} .
\end{equation} 
Although this generator is not uniformly elliptic (there is a second derivative in the momentum variable but not in the position) the well-posedness of the Langevin equation has been extensively studied and in fact for this harmonic case one can explicitly find the solution (for more information see for instance \cite{Pav14}, Chapter 6).

We choose to study the canonical process \eqref{eq:cano} obtained from \eqref{eq:langevin}
 when biasing with respect to the three-dimensional vector observable 
\begin{equation}\label{eq:3dobs}
A_{T}=\left(\int_{0}^{T}\left(P_{1,t}\right)^{2} dt \,,\int_{0}^{T}\left(P_{2,t}\right)^{2} dt \,,\int_{0}^{T}P_{1,t}P_{2,t}  \,dt\right) .
\end{equation}
In order to find the scaled cumulant generating function $\Lambda_0(k)$ we compute the largest eigenvalue of a tilted generator $L_k$ (see for instance \cite{CT15a} for its construction) 
\begin{equation}\label{eq:tilted}
    L_k = L + k_1 p_1^2 + k_2 p_2^2 + k_3 p_1 p_2
\end{equation}
where $k=(k_1,k_2,k_3) \in \mathbb{R}^3$ is the vector of tilting parameters. The bijection $R \equiv R_\theta$ is chosen to be a rotation of an angle $\theta$ in the 2D momentum space.
The observable $A_T$ transforms under $R_\theta$ as $A_T(R_\theta \omega) = \mathcal{R}_\theta A_T(\omega)$ with 
\begin{equation}
    \mathcal{R}_\theta = 
    \begin{pmatrix}
        \cos^2(\theta) & \sin^2(\theta) & -2 \sin(\theta)\cos(\theta) \\
        \sin^2(\theta) & \cos^2(\theta) & 2 \sin(\theta)\cos(\theta) \\
        \sin(\theta)\cos(\theta) & -\sin(\theta)\cos(\theta) & \cos^2(\theta) - \sin^2(\theta).
    \end{pmatrix}
    \label{eq:Rot}
\end{equation}
Note that this matrix is not orthogonal, namely $\mathcal{R}_\theta^\dag \neq \mathcal{R}_\theta^{-1}$, despite the fact that $\det(\mathcal{R}_\theta)=1$, as one can easily check. In order to find the leading eigenvalue and corresponding eigenvector of the tilted generator \eqref{eq:tilted}, one can take advantage of the results derived in \cite{dBT23} for linear diffusions and quadratic observables (see also the rigorous mathematical analysis developed in \cite{JPS17} for entropic observables). In particular, given a linear SDE with additive noise
\begin{equation}
    d X_t = -M X_t dt + \sigma d W_t,
\end{equation}
whose generator is 
\begin{equation*}
    L= -M x \cdot \nabla  + \frac{1}{2} \sigma \sigma^\dag \Delta ,
\end{equation*}
and a quadratic observable $$A_T = \int_0^T \langle X_t, Q X_t \rangle dt,$$ it was shown that the scaled cumulant generating function can be retrieved from the dominant eigenvalue of a tilted generator $L_k = L + k \langle x, Q x \rangle$
\begin{equation}
    L_k r_k = \Lambda_0(k) r_k, \quad r_k (x)  = \mathrm{e}^{\langle x, B_k x \rangle},
\end{equation}
that, given the Gaussian ansatz for $r_k$, explicitly reads
\begin{equation}
    \Lambda_0(k)= \mathrm{Tr}(D B_k) ,
\end{equation}
where $D=\sigma \sigma^{\dag}$ and the matrix $B_k$ satisfies the algebraic Riccati equation
\begin{equation}\label{eq:RicPRE}
    2 B_k D B_k - M^\dag B_k - B_k M + kQ =0.
\end{equation}
This can be easily checked applying the tilted generator to the  Gaussian ansatz for the right eigenvector. In particular, the term $kQ$ in the Riccati equation arises from the term $k\langle x, Q x \rangle$ in the tilted generator.
This can be immediately adapted to our case. Indeed, the only difference is that, dealing with vector-valued observables, one needs to tilt each component with a different variable. However, the additional biasing term $k_1 p_1^2 + k_2 p_2^2 + k_3 p_1 p_2$ in \eqref{eq:tilted} can be conveniently written as $\langle x, \widetilde{Q}_k x \rangle$, with $x=(q_1,p_1,q_2,p_2)\in \mathbb{R}^4$ and 
\begin{equation}
    \widetilde{Q}_k = 
    \begin{pmatrix}
       0 & 0 & 0 & 0 \\
       0 & k_1 &0 &\frac{k_3}{2} \\
       0 &0 & 0 & 0 \\
       0 &\frac{k_3}{2} & 0 & k_2
    \end{pmatrix} = \begin{pmatrix}
        k_1 & \frac{k_3}{2} \\
        \frac{k_3}{2} & k_2
    \end{pmatrix}
    \otimes
    \begin{pmatrix}
        0 & 0 \\
        0 & 1
    \end{pmatrix} ,
\end{equation}
where the Kronecker product\footnote{The Kronecker product $A \otimes B$ of a $n\times m$ matrix $A$ and a $p \times q$ matrix $B$ is the $pn \times qm$ block matrix $A \otimes B = 
\begin{pmatrix}
 a_{11} B & \ldots & a_{1m} B  \\
 \vdots & \ddots & \vdots \\
 a_{n1} B & \ldots & a_{nm} B
\end{pmatrix}.$} $\otimes$ has been used in the second equality.
Therefore, the algebraic Riccati equation relevant to our problem reads
\begin{equation}\label{eq:Ric}
    2 B_k D B_k - M^\dag B_k - B_k M + \widetilde{Q}_k =0,
\end{equation}
where, in particular,
\begin{align}
    M = &
    \begin{pmatrix}
       0 & -1 & 0 & 0 \\
       1 & \gamma &0 &0 \\
       0 &0 & 0 & -1 \\
       0 &0 & 1 & \gamma
    \end{pmatrix} = \bbbone_2 \otimes 
    \begin{pmatrix}
       0 & -1 \\
       1 & \gamma 
    \end{pmatrix} , \quad 
    \\
    \label{eq:D}
    D = &
    \begin{pmatrix}
       0 & 0 & 0 & 0 \\
       0 & \frac{2\gamma}{\beta} &0 &0 \\
       0 &0 & 0 & 0 \\
       0 &0 & 0 & \frac{2\gamma}{\beta} 
    \end{pmatrix} = \bbbone_2 \otimes 
    \begin{pmatrix}
       0 & 0 \\
       0 & \frac{2\gamma}{\beta} 
    \end{pmatrix} . 
\end{align}
The equation \eqref{eq:Ric} in general has many solutions. The interesting one for us is the one that satisfies $B_0=0$, because $r_0(x) = 1$ has to hold for the original (untilted) Markov generator. 
The equation is solved explicitly in Appendix \ref{app_langevin} for an open set $S$ of values of $k$ containing the origin
\begin{equation}
    S= \left\{ k=(k_1,k_2,k_3) \in \mathbb{R}^3 \Big| \, k_1+k_2 \leq \frac{\gamma \beta}{2} \; \wedge \; k_3^2 - 4k_1 k_2 \leq \frac{\beta^2\gamma^2}{4} - \beta\gamma (k_1 + k_2)\right\}.
\end{equation}
The symbol $\wedge$ indicates here the logical operator \emph{and}, meaning that both inequalities need to be satisfied.
Explicitly, the scaled cumulant generating function reads
\begin{equation}\label{eq:Lambda0_langevin}
    \Lambda_0(k) =  \gamma - \gamma\sqrt{  \left( \frac{1}{2} - \frac{k_1 +k_2}{\beta\gamma}   \right) + \sqrt{\left( \frac{1}{2} - \frac{k_1 +k_2}{\beta\gamma}   \right)^2 - \frac{1}{\beta^2\gamma^2}\Big[ k_3^2 + (k_1 -k_2)^2 \Big]}},
\end{equation}
for any $k \in S$ while it is infinite for $k \notin S$. The second inequality defining the set $S$ corresponds to the positivity of the term in the most internal square root, so that the expression in \eqref{eq:Lambda0_langevin} gives indeed a real number. It is also checked in Appendix \ref{app_langevin} that $\Lambda_0(k)$ is steep at the boundary, and therefore essentially smooth. This guarantees that G\"artner-Ellis theorem applies.

We now check that the symmetry holds true, namely $\Lambda_0(k)= \Lambda_0(\mathcal{R}_\theta^\dag k)$ for any $\theta$ and $k \in S$. This can be done noting that the two quantities $C_1(k)= k_1 + k_2$ and $C_2(k)= (k_1 - k_2)^2 + k_3^2$ are invariant under the action of $\mathcal{R}_\theta^\dag $. Indeed, given $\overline{k}= \mathcal{R}_\theta^\dag k$ one has
\begin{align*}
    &\overline{k}_1 = \cos^2(\theta)   k_1 + \sin^2(\theta) k_2 + \sin(\theta)\cos(\theta) k_3 ,\\
    &\overline{k}_2 = \sin^2(\theta)   k_1 + \cos^2(\theta) k_2 - \sin(\theta)\cos(\theta) k_3 ,\\
    &\overline{k}_3 = - \sin(2\theta)  k_1 + \sin(2\theta) k_2 + \cos(2\theta) k_3 ,
\end{align*}
and therefore one can easily verify that $\overline{k}_1 + \overline{k}_2 = k_1 + k_2$ and $(\overline{k}_1 - \overline{k}_2)^2 + \overline{k}_3^2 = (k_1 - k_2)^2 + k_3^2$. 
Therefore, $\Lambda_0(k)$ is invariant because the $k$-dependence can be written only in terms of $C_1$ and $C_2$. Also, the conditions defining $S$ only depend on $k$ through $C_1$ and $C_2$, so that also the set of allowed $k$ values is mapped into itself by $\mathcal{R}_\theta$ for any $\theta$, i.e. $\mathcal{R}_\theta^\dag k \in S$ if $k \in S$. 

The rate function can also be computed explicitly (see Appendix \ref{app_langevin}) and reads
\begin{equation}\label{eq:Ilangevin}
    I_0(a) = -\gamma + \frac{\gamma \beta}{4}(a_1+a_2) + \frac{\gamma}{4\beta}\, \frac{a_1+a_2}{a_1 a_2 - a_3^2} , \qquad a_1,a_2 >0 ,\quad a_3^2 < a_1 a_2.
\end{equation}
As for the cumulant generating function, in order to check the symmetry $I(\mathcal{R}_\theta a)= I(a)$ it is useful to notice that $I(a)$ is written in terms of invariant quantities. In fact, one can easily check that both $a_1 + a_2$ and $a_1 a_2 - a_3^2$ are invariant under the action of $\mathcal{R}_\theta$. Since $\mathcal{R}_\theta \neq \mathcal{R}_\theta^\dag$, the second quantity would not be invariant under $\mathcal{R}_\theta^\dag$ and one has to be careful in using the right symmetry condition for the rate function and the cumulant generating function.

We now consider the canonical process constraining the possible values of the field $s$ to be in the set $S$. For a fixed $S$, we then know that the relation \eqref{eq:jolieRF}  $\Lambda_s(k) = \Lambda_0 (k+s) - \Lambda_0 (s)$ holds true for those values of $k$ in a set $\widetilde{S}_s \equiv \{ k | \, (k+s) \in S \}$. Therefore, in this domain, the scaled cumulant generating function $\Lambda_s(k)$
\begin{align}
    \Lambda_s(k)=& - \gamma\sqrt{  \left( \frac{1}{2} - \frac{C_1(k+s)}{\beta\gamma}    \right) + \sqrt{\left( \frac{1}{2} - \frac{C_1(k+s)}{\beta\gamma}    \right)^2 - \frac{C_2(k+s)}{\beta^2\gamma^2}}} \nonumber \\
    &+ \gamma\sqrt{  \left( \frac{1}{2} - \frac{C_1(s)}{\beta\gamma}    \right) + \sqrt{\left( \frac{1}{2} - \frac{C_1(s)}{\beta\gamma}    \right)^2 - \frac{C_2(s)}{\beta^2\gamma^2}}}   
\end{align}
satisfies the FR \eqref{eq:MR-1-1} that is inherited from the symmetry of $\Lambda_0$. Explicitly,
\begin{equation}
    \Lambda_s(k) = \Lambda_s(\overline{k}+\overline{s}-s),
\end{equation}
where $\overline{k}$ is the vector of components $\overline{k}_1,\overline{k}_2,\overline{k}_3$ and $\overline{s}$ reads in the same way when $k$ is substituted with $s$. 

For a continuous family of transformations we can also look at the relations on cumulants. 
The infinitesimal generator of  $\mathcal{R}\left(\theta\right)$ given in  \eqref{eq:Rot} is in this case 
\[
\mathcal{N}=\left(\begin{array}{ccc}
0 & 0 & -2\\
0 & 0 & 2\\
1 & -1 & 0
\end{array}\right)
\]
and for example, the $p=1$ relation \eqref{eq:Super1} reads
\begin{equation}
\sum_{j_{1}=1}^{3}\mathcal{N}_{i_{1}j_{1}}\mathbb{E}_{\mathbb{P}^{cano,s}}\left(A_{T,j_{1}}\right)=-\sum_{j_{2}=1}^{3}\sum_{j_{3}=1}^{3}s_{j_{3}}\mathcal{N}_{j_{3}j_{2}}\mathbb{E}_{\mathbb{P}^{cano,s}}^{c}\left(A_{T,i_{1}}A_{T,j_{2}}\right),
\end{equation}
for any $ i_1 \in \{1,2,3 \} $.
Explicitly, e.g. for $i_1=1$ we have
\begin{equation*}
{E}_{\mathbb{P}^{cano,s}}\left(A_{T,3}\right)= (s_2-s_1){E}_{\mathbb{P}^{cano,s}}^c \left(A_{T,1}A_{T,3}\right) + \frac{s_3}{2} \Big( {E}^c_{\mathbb{P}^{cano,s}} \left((A_{T,1})^2\right) - {E}^c_{\mathbb{P}^{cano,s}} \left(A_{T,1}A_{T,2}\right) \Big),
\end{equation*}
that for the specific observable considered here means
\begin{align}
   {E}_{\mathbb{P}^{cano,s}}\left(\int_0^T P_{1,t}P_{2,t} dt\right) =& (s_2-s_1){E}_{\mathbb{P}^{cano,s}}^c \left(\int_0^T P^2_{1,u} du\int_0^T P_{1,t}P_{2,t} dt\right) \nonumber\\
   &+ \frac{s_3}{2}  {E}^c_{\mathbb{P}^{cano,s}} \left(\int_0^T P^2_{1,u} du \int_0^T P^2_{1,t} dt\right) \nonumber\\
   &- \frac{s_3}{2} {E}^c_{\mathbb{P}^{cano,s}} \left(\int_0^T P^2_{1,u} du \int_0^T P^2_{2,t} dt\right) .
\end{align}
The physical content of these kind of relations needs to be further explored and better understood with more examples.

\section{Case study 2: FRs for non-degenerate diffusion processes associated to a time-local transformation}
\label{sec:Markovian}

\subsection{General Setting and Results}
In this Section, we consider a Markov process $X_{t}$ with state space $E= \mathbb{R}^n$ solving the following SDE:
\begin{equation}\label{eq:diff}
dX_{t}=F(X_{t})dt+ \sigma(X_{t})\circ dW_t,
\end{equation}
where $F: \mathbb{R}^n \to \mathbb{R}^n$ is a vector field, $\sigma$ is a $n\times m$ matrix function and $W_t$ is a $m$-dimensional vector of independent Wiener processes. $F$ and $\sigma$ are chosen such that a strong solution exists (see for instance \cite{Kha12} Theorems 3.4 and 3.5 for sufficient conditions) but are arbitrary otherwise. The symbol $\circ$ indicates that the Stratonovich-Fisk convention \cite{Fis63,Str68} is used. The initial condition is sampled according to a probability density $\varrho(x)$. The explicit form of the generator is then 
\begin{equation}
L=\widehat{F} \cdot \nabla+\frac{1}{2}\nabla \cdot D \, \nabla,\label{eq:gpd}
\end{equation}
with the covariance $D \equiv \sigma \sigma^\dag$ and the modified drift $\widehat{F}: \mathbb{R}^n \to \mathbb{R}^n$ whose components explicitly read
\begin{equation}\label{eq:moddrift}
\widehat{F}_i(x)=F_i(x)- \widetilde{F}_i(x), \quad \widetilde{F}_i(x)\equiv \frac{1}{2}\sum_{jk}\sigma_{ik}(x)\frac{\partial \sigma_{jk}(x)}{\partial x_j} \; .
\end{equation}
The non-degeneracy condition means that the matrix $D(x)$ is strictly positive ($D(x)>0$), and therefore invertible. The diffusion process studied in the previous example \eqref{eq:D} does not belong to this class.

We consider in the following the particular case of bijections $R$ that act pointwise in time, namely the following assumption is true:
\begin{assumption}[Local transformation]\label{ass:local}
 There exists an invertible map $r : \mathbb{R}^n \to \mathbb{R}^n$ such that $\left(R\left(X_{\left[0,T\right]}\right)\right)_{t}=r\left(X_{t}\right)$ for any $t \in [0,T]$.    
\end{assumption}

In this case, the transformed process $r(X_t)$ satisfies a modified SDE and the path measure $\mathbb{P}_{\left[0,T\right]}\circ R$ also describes a Markovian diffusion process. In particular, the following lemma holds true \cite{CT25+}: 
\begin{lemma}[Modified diffusion]  
 Given a transformation $R$ and a map $r$ satisfying Assumption \ref{ass:local}, the path measure $\mathbb{P}_{\left[0,T\right]}\circ R$ also describes a Markovian diffusion process, with drift $F^{R}$ and diffusion matrix $\sigma^{R}$ given by
\begin{equation}\label{eq:imdrift-1-1}
\begin{cases}
 F^{R}(x)\equiv\left(J_{r}F\right)\left(r^{-1}\left(x\right)\right)\\
 \sigma^{R}(x) \equiv \left(J_{r}\sigma\right)\left(r^{-1}\left(x\right)\right)
\end{cases},
\end{equation}
where $J_{r}(x)$ is the jacobian matrix of the $r$ transform: 
\begin{equation}\label{eq:Jac-2}
\left(J_{r}(x)\right)_{ij}\equiv\frac{\partial r^{i}}{\partial x^{j}}(x) \;. 
\end{equation}
In particular, the modified covariance matrix $D^{R} \equiv \sigma^{R}(\sigma^{R})^\dag$ reads
\begin{equation}\label{eq:newcov}
    D^{R}(x) =\left(J_{r} D J_{r}^{\dag}\right)\left(r^{-1}\left(x\right)\right)  \;.
\end{equation}
\end{lemma}
\begin{proof}
    We always assume the Stratonovich convention so that the standard rules of calculus apply. In particular, defining $X_t' \equiv r(X_t)$ one has $d X_t'= J_r(X_t) \circ d X_t$. Then, one can use the SDE \eqref{eq:diff} to substitute $d X_t$, obtaining
    \begin{equation*}
        d X_t'= J_r(X_t) F(X_t) dt +  J_r(X_t) \sigma(X_{t})\circ dW_t.
    \end{equation*}
    Finally, since $r$ is invertible, one can use $X_t= r^{-1}(X_t')$ in the right hand side of the equation and find an SDE for $X_t'$. The resulting drift and diffusion matrix are the ones stated in \eqref{eq:imdrift-1-1}. The formula \eqref{eq:newcov} for the covariance matrix then follows from the definition $D \equiv \sigma \sigma^\dag$.
\end{proof}
Given the previous result, the problem of comparing the two path measures $\mathbb{P}_{\left[0,T\right]}$ and $\mathbb{P}_{\left[0,T\right]}\circ R$ reduces to comparing the path measures of two different diffusion processes. First of all, we should assume the two measures to be absolutely continuous with respect to each other, so that the Radon-Nikodym derivative (and therefore the scalar entropic functional) defined in Eq.~\eqref{eq:AF} exists.
By Girsanov lemma \cite{SV79}, in the case where $D(x)>0$ (that is our set-up), the scalar entropic functional exists if and only if the covariance $D$ remains unchanged, i.e. $D^R=D$, or equivalently for any $x$
\begin{equation}\label{eq:Dinvariant}
D\left(r\left(x\right)\right)\equiv\left(J_{r} D J_{r}^\dag\right)(x).
\end{equation}
Moreover, since we assumed a strong solution exists to the SDE \eqref{eq:diff}, in this case one can compute explicitly $\varSigma_{T}^{\mathbb{P}}$ by means of the Cameron-Martin-Girsanov theorem \cite{SV79}. Since the modified diffusion process has a generator of the form 
\begin{align}
    L^{R} &= \widehat{F}^R \cdot \nabla + \frac{1}{2} \nabla \cdot D^{R}\, \nabla \nonumber\\
    &= \widehat{F}^R \cdot \nabla + \frac{1}{2} \nabla \cdot D \nabla \;,
\end{align}
where $\widehat{F}^R$ is the modified drift obtained as in \eqref{eq:moddrift} with $F^R$ replacing $F$ and $\sigma^R(x) \equiv \left(J_r \sigma \right)\left(r^{-1}(x)\right)$ replacing $\sigma(x)$, one has that $L-L^R = (\widehat{F}^R-\widehat{F}) \cdot \nabla$ and therefore (see for instance Appendix A in \cite{CT15a})
\begin{align}\label{eq:girsanov}
    \varSigma_{T}^{\mathbb{P}}=& -\int_{0}^{T} \left\langle \left(\widehat{F}^R-\widehat{F}\right)(X_{t}),D^{-1}(X_{t})\circ dX_{t}\right\rangle + \ln\left(\frac{\varrho(X_0)}{\varrho(rX_0)}\right) \nonumber \\
    & + \frac{1}{2}\int_{0}^{T} \Big( \left\langle  \left(\widehat{F}^R -\widehat{F}\right)(X_{t}),D^{-1}(X_{t})\left(\widehat{F}^R+\widehat{F}\right)(X_{t})\right\rangle + \nabla \cdot \left(\widehat{F}^R- \widehat{F}\right)(X_{t}) \Big) dt \;,
\end{align}

In general, it may be not possible to recast this expression in a way that satisfies Assumption \ref{ass}. In order to proceed, we make further hypotheses on the structure of the process.

\begin{assumption}[Symmetry properties]\label{ass:symmetry} The following properties hold true:
\begin{itemize}
    \item[(i)] the entropic functional exists, i.e. one has \eqref{eq:Dinvariant} 
    \item[(ii)] the probability density $\varrho(x)$ of the initial condition is invariant under the transformation $r$, i.e $\varrho(x)=\varrho(rx)$ for any $x \in \mathbb{R}^d$
    \item[(iii)] the vector field $F(x)$ is the sum of a constant term $\overline{F}$ and a symmetric vector field $F_S(x)$, such that $F_S(x)= J_r F_S (r^{-1}(x))$
    \item[(iv)] the modification of the drift is such that $\widetilde{F}^R(x)=\widetilde{F}(x)$, namely for all $i$ 
    $$ \sum_{jk}\sigma^R_{ik}(x)\frac{\partial \sigma^R_{jk}(x)}{\partial x_j} = \sum_{jk}\sigma_{ik}(x)\frac{\partial \sigma_{jk}(x)}{\partial x_j}\,,$$
    and it is symmetric, i.e. $\widetilde{F}(x)= J_r \widetilde{F}(r^{-1}(x))$
    \item[(v)]  the covariance satisfies $D(r(x))= D(x)$
    \item[(vi)] the jacobian is space-time homogeneous $J_r(x) \equiv J_r$ (this is the case when $r$ is some affine transformation)
\end{itemize}
\end{assumption}

This assumption seems in fact very restrictive. Note however that we use it to treat multiplicative noise. In the case of additive noise, when $\sigma$ is a constant, the modified drift equals the original drift, so that $\widetilde{F}(x)=\widetilde{F}^R(x)=0$ and \emph{(iv)} is trivially satisfied. Also, \emph{(v)} is trivially satisfied because $D$ is constant. Therefore, in the additive case one assumes some symmetry in the drift (allowing for a constant asymmetric part) and in the initial condition, while for general multiplicative noise more constraints on the noise matrix are needed. Indeed, the conditions on $\widehat{F}^R(x)$ and $\widetilde{F}(x)$ are in fact conditions on $\sigma^R(x)$ and $\sigma(x)$. In particular, they imply that $\widehat{F}^R(x)-\widehat{F}(x)= F^R(x)-F(x)$ and that $\widehat{F}^R(x)+\widehat{F}(x)$ is symmetric. We will consider later an example when these conditions are all verified and still the dynamics is nontrivial.

We can now state the main result of this Section.

\begin{theorem}[Finite time FRs for diffusions]
Let $\mathbb{P}_{\left[0,T\right]}$ be the path measure associated with the solution of the nondegenerate SDE \eqref{eq:diff}. Let $R$ be a bijection in the space of trajectories $\Omega_T=C([0,T], \mathbb{R}^n)$ that satisfies Assumption \ref{ass:local}. The vector field $F$ and the diffusion matrix $\sigma$ satisfy Assumption \ref{ass:symmetry}. Then, the Fluctuation Relations \eqref{eq:MRF-2-1-1} and \eqref{eq:beau} hold true with the following triple $(A_T, w, \mathcal{R})$
    $$\begin{cases}
     A_T=\int_{0}^{T}  D^{-1}(X_{t})\circ \Big(dX_{t} - \frac{1}{2} S(X_t) dt \Big) \\
     w = (\bbbone -J_r) \overline{F} \\
     \mathcal{R} = (J_r^\dag)^{-1}
    \end{cases}$$
where $S(x)$ is the vector field with components $ S_i(x)= 2 F_{S,i}(x) - \sum_{jk}\sigma_{ik}(x)\frac{\partial \sigma_{jk}(x)}{\partial x_j}$. Explicitly, they are
\begin{equation}\label{eq:diffFR1}
\frac{d\mathcal{P}_{A_T}}{d\mathcal{P}_{(J_r^\dag)^{-1} A_T}}(a)=\exp \Big\langle (\bbbone -J_r)\overline{F},(J_r^\dag) a \Big\rangle  \; ,
\end{equation}
\begin{equation}\label{eq:diffFR2}
\mathbb{E}\Big[\exp\Big\langle k,A_T\Big\rangle \Big]=\mathbb{E}\Big[\exp \Big\langle (J_r)^{-1} k -
(\bbbone-J_r)\overline{F}, \,A_T\Big\rangle \Big] .
\end{equation}
\end{theorem}

\begin{proof}
    Due to Assumption \ref{ass:symmetry}, the vector fields $F(x)$ and $F^R(x)$ read
    \begin{align*}
        F(x) &= \overline{F} + F_S(x)  \\
        F^R(x) &= J_r \overline{F} + F_S(x) \,,
    \end{align*}
    so that the difference turns out to be constant $(F^R- F)(x)= (\widehat{F}^R- \widehat{F})(x)= (J_r - \bbbone)\overline{F}$
    and the last term in Eq. \eqref{eq:girsanov} is vanishing. Consider now the first term in the second line of \eqref{eq:girsanov}. The sum $(\widehat{F}^R+\widehat{F})(x)$ consists of a constant term $(J_r + \bbbone)\overline{F}$ plus the term $S(x)$ which is symmetric.
    We now show that the constant term does not contribute. Indeed, with our assumptions $J_r D (x) J_r^\dag =D(x) = J_r^{-1} D (x) (J_r^\dag)^{-1}$ and therefore one also has $D^{-1}(x)=  J_r^\dag D^{-1}(x) J_r$. Computing the scalar product, we consequently get
    \begin{align*}
       &\left\langle  (J_r - \bbbone)\overline{F} ,D^{-1}(X_{t}) (J_r + \bbbone)\overline{F} \right\rangle \\
       =& \left\langle  J_r \overline{F} ,D^{-1}(X_{t}) \, J_r  \overline{F} \right\rangle  - \left\langle  \overline{F} , D^{-1}(X_{t})  \,\overline{F} \right\rangle \\
       =& \left\langle  \overline{F} , \Big( J_r^\dag D^{-1}(X_{t}) J_r -  D^{-1}(X_{t}) \Big)\overline{F} \right\rangle = 0 \,.
    \end{align*}
    As a result, the entropic functional can be written as follows
    \begin{equation*}
       \varSigma_{T}^{\mathbb{P}}=  \left\langle (\bbbone -J_r) \overline{F}, \int_{0}^{T}  D^{-1}(X_{t})\circ \Big(dX_{t} - \frac{1}{2} S(X_t) dt \Big) \right\rangle \,,
       \label{eq:GenStefano}
    \end{equation*}
    thus reading as in Eq. \eqref{eq:ass1}, for instance with 
    $$\begin{cases}
     A_T\left(X_{\left[0,T\right]}\right)=\int_{0}^{T}  D^{-1}(X_{t})\circ \Big(dX_{t} - \frac{1}{2} S(X_t) dt \Big) \\
     w= (\bbbone -J_r) \overline{F} 
    \end{cases}.$$ 
    In order to conclude the proof, one has to check that this observable $A_T\left(X_{\left[0,T\right]}\right)$ does indeed satisfy covariance property given in Eq. \eqref{eq:ass1}. This is readily done, putting together the different properties that follow from Assumption \ref{ass:symmetry}. In particular, one has
    \begin{align*}
        &D^{-1}(r(X_t))= D^{-1}(X_t), \quad  d r(X_t)= J_r dX_t, \\
        &S(r(X_t))= J_r S(X_t),\quad  D^{-1}(X_t) J_r= (J_r^\dag)^{-1} D^{-1}(X_t) ,
    \end{align*}
    so that in the end $A\left(R X_{\left[0,T\right]}\right)= (J_r^\dag)^{-1} A\left(X_{\left[0,T\right]}\right)$.    
\end{proof}
We can now make a few remarks on the result. \\
(I) First of all, if $r$ is an involution, the drift $F$ can be always decomposed as $F=F_A + F_S$ with
\begin{equation}
    \begin{cases}
     F_{S}=\frac{F+ F^R}{2}\\
     F_{A}=\frac{F- F^R}{2}
    \end{cases},
\end{equation}
where $F_{S}$ is symmetric under $r$ (i.e. $ J_r F_{S}(r^{-1}x)=F_{S}(x)$)
and $F_{A}$ is anti-symmetric under $r$ (i.e. $J_r F_{A}(r^{-1}x)=-F_{A}(x)$). Therefore, in this case, the only assumption one needs on $F$ is that $F_A$ is a constant. \\
(II) The decomposition of the entropic functional is not unique. For instance, one could use
$$\begin{cases}
     A_T\left(X_{\left[0,T\right]}\right)=\int_{0}^{T} (\bbbone -J_r^\dag) D^{-1}(X_{t})\circ \Big(dX_{t} - \frac{1}{2} S(X_t) dt \Big) \\
     w=  \overline{F} 
\end{cases},$$ 
and everything holds true with the same $\mathcal{R}= (J_r^\dag)^{-1} $. Also, if the noise is additive, i.e. $D(x)\equiv D$, one can write
$$\begin{cases}
     A_T\left(X_{\left[0,T\right]}\right)= X_T -X_0 - \frac{1}{2}\int_{0}^{T} S(X_t) dt  \\
     w= D^{-1} (\bbbone -J_r^\dag) \overline{F} 
\end{cases}.$$
\\ 
(III) Despite the lack of uniqueness, the freedom in the choice of the observable $A_T$ is little, contrary to the canonical path probability set-up of the previous section, where the Assumption \ref{ass_cano} was the only restriction.  
\\
(IV) The proof is very similar to the one presented for the \emph{spatial fluctuation relations} of \cite{PRG15} and essentially based on the Cameron-Martin-Girsanov formula plus symmetry assumptions on the process. However, we consider here a more general situation and we are able to treat also a class of diffusions with multiplicative noise.

Concerning the behavior for long times, it is more challenging in this case to find sufficient conditions such that Assumption \ref{ass:ld} is valid and therefore our main result Eq. \eqref{eq:MR} holds true. Based on the finite-time result and heuristic considerations about large deviations, the expected \emph{Asymptotic Fluctuation Relation} should read
\begin{equation}\label{eq:MR-2}
\begin{cases}
 I_{\mathbb{P}}\left(a\right)=\left\langle (\bbbone -J_r) \overline{F},J_{r}^\dag a\right\rangle +I_{\mathbb{P}}\left(J_{r}^\dag a\right)\\
 \Lambda_{\mathbb{P}}\left(k\right)=\Lambda_{\mathbb{P}}\left(J_{r}^{-1} k- (\bbbone -J_r) \overline{F}\right) 
\end{cases},
\end{equation}
or equivalently,
\begin{equation}\label{eq:MR-2-1}
\begin{cases}
I_{\mathbb{P}}\Big((J_r^\dag)^{-1} a\Big)=\Big\langle (\bbbone -J_r) \overline{F}, a\Big\rangle +I_{\mathbb{P}}\big(a\big)\\ \Lambda_{\mathbb{P}}\left(J_{r} k\right)=\Lambda_{\mathbb{P}}\left(k - (\bbbone -J_r) \overline{F}\right) 
\end{cases}.
\end{equation} 
For a system with compact state space, as done in \cite{LS99} in the involution case, starting from the Girsanov formula one can show that the limit 
\begin{equation}
    \Lambda_{\mathbb{P}}(k) = \lim_{t \to \infty} \frac{1}{T} \ln \mathbb{E}_{\mathbb{P}}\left[\exp\Big\langle k,A_T\Big\rangle \right]
\end{equation}
exists for $k$ in some open ball of $\mathbb{R}^n$ containing the origin, using Perron-Frobenius type arguments. Then, G\"artner-Ellis theorem gives the large deviation principle for the family of probabilities $\mathcal{P}_{A_T}=\mathbb{P} \circ A_T^{-1}$ with rate functional $I_{\mathbb{P}}$, thus completing the requirements of Assumption \ref{ass:ld}. In such a case the \emph{Asymptotic Fluctuation Relation} \eqref{eq:MR-2}-\eqref{eq:MR-2-1} for $A_T$ can be proved rigorously.

When the state space is not compact, like for instance $\mathbb{R}^n$ in the present case, the previous strategy in general fails, as pointed out in \cite{JPS17}, because the observable and even the temporal boundary terms are typically unbounded, so that one has to resort to ad hoc methods. For instance, in \cite{BGL23} the authors proved a large deviation principle for a suitable modification of the entropic functional (removing unbounded contributions), using the contraction technique from a so-called \emph{level 3} large deviation result \cite{BGL23,Var84}. According to the general Donsker-Varadhan theory for the long-time asymptotics of Markov processes \cite{DV75a,DV75b,DV76,DV83}, one can define the empirical process as the map\footnote{The empirical process is usually denoted $R_T$ in the literature, but we choose here a different letter so as to avoid confusion with the transformation $R$.} $\zeta_T: C([0,T],\mathbb{R}^n) \to \mathcal{P}$, 
\begin{equation}
    \zeta_T(X) := \frac{1}{T}\int_0^T \delta_{\theta_t X^T} dt
\end{equation}
where $\theta_t$ is the time-translation operator, $\mathcal{P}$ is the set of translation-invariant probabilities on $D(\mathbb{R}, \mathbb{R}^n)$ and $(X^T)_t := X_{t- \lfloor t/T \rfloor T}$
is the $T$-periodization of $X_t$. The level 3 large deviation principle is a large deviation principle pertaining to the statistics of the empirical process and lower level large deviation principle should then be obtained by contraction. More similar to the general strategy mentioned before in the paper is the approach of \cite{BdG15} and \cite{Raq24} where the authors relate the scaled cumulant generating function to the dominant eigenvalue of a differential operator (a modified generator) and by means of a careful analysis of the domain issues they can reduce the problem to the one where (some variation of) the G\"artner-Ellis theorem applies.

\subsection{Example: Rotationally-invariant multiplicative noise}

To illustrate the novelty of our result we present here a simple example of a diffusion with \emph{multiplicative} noise that satisfies the FRs.
Let us consider the SDE \eqref{eq:diff} for a two-dimensional stochastic process $X_t \in \mathbb{R}^2$ with initial condition $X_0=0$. We choose the linear transformation $r_\theta (x) \equiv U_\theta x$ as a rotation parametrized by the angle $\theta$
\begin{equation}
U_\theta = 
\begin{pmatrix}
\cos(\theta)  &-\sin(\theta) \\
\sin(\theta)   & \cos(\theta) 
\end{pmatrix}.
\end{equation}
In this case one has $J_r = U_\theta$ independently of $x$, that is point $(vi)$ in Assumption \ref{ass:symmetry}. The condition $X_0=0$ implies that point $(ii)$ in Assumption \ref{ass:symmetry} holds true as well.
To satisfy the other symmetry constraints, we want the diffusion matrix to be rotationally invariant, but at the same time we want to allow for a multiplicative noise in order to make the example more interesting. A possibility is to choose 
\begin{equation}
    \sigma(x) = f(|X_t|^2) 
    \begin{pmatrix}
\cos(\lambda)  &-\sin(\lambda) \\
\sin(\lambda)   & \cos(\lambda) 
    \end{pmatrix}
\end{equation}
where $\lambda$ is a real parameter and $f: \mathbb{R}_+ \to \mathbb{R}$ is a scalar function of $|X_t|^2$ such that $f(x)\neq 0$ for all $x \in \mathbb{R}_+$. In this case, the covariance matrix explicitly reads
\begin{equation*}
    D(X_t) = f^2(|X_t|^2) \bbbone_2 ,
\end{equation*}
and one immediately has $D(r(X_t))=J_r D J_r^\dag (X_t)=D(X_t)$, so that the points $(i)$ and $(v)$ of Assumption \ref{ass:symmetry} hold true. We also assume the drift to be $F(X_t) = \overline{F} + F_S(X_t)$ in order to cope with point $(iii)$. The modified drift $\widehat{F}$ can also be computed readily from \eqref{eq:moddrift} (we use here the notation $\sigma_{ij}(x)= f(|x|^2)\sigma_{ij}$)
\begin{align}
    \widehat{F}_1(x) &= F_1(x) - \frac{1}{2} f(|x|^2) \Big( \frac{\partial f(|x|^2)}{\partial x_1} (\sigma_{11}^2 + \sigma_{12}^2) + \frac{\partial f(|x|^2)}{\partial x_2} (\sigma_{11}\sigma_{21} + \sigma_{12}\sigma_{22})  \Big) \nonumber \\
    &= F_1(x) - \frac{1}{2} f(|x|^2) \frac{\partial f(|x|^2)}{\partial x_1} \;, \\
    \widehat{F}_2(x) &= F_2(x) - \frac{1}{2} f(|x|^2) \Big( \frac{\partial f(|x|^2)}{\partial x_1} (\sigma_{21}\sigma_{11} + \sigma_{22}\sigma_{12}) + \frac{\partial f(|x|^2)}{\partial x_2} (\sigma_{21}^2 + \sigma_{22}^2)  \Big) \nonumber \\
    &= F_2(x) - \frac{1}{2} f(|x|^2)  \frac{\partial f(|x|^2)}{\partial x_2} \;.
\end{align}
More compactly, one can write $\widehat{F}(X_t) = \overline{F} + F_S(X_t) - \left(f\cdot f'\right) (|X_t|^2) X_t$. We need to compute the drift in the rotated process as well. It is $F^R(x)= J_r \overline{F} + F_S(x)$. The diffusion matrix reads \eqref{eq:imdrift-1-1}
$$\sigma^R(x)= J_r \sigma(r^{-1}x)= J_r \sigma(x) = f(|x|^2)
\begin{pmatrix}
\cos(\lambda +\theta)  &-\sin(\lambda + \theta) \\
\sin(\lambda + \theta)   & \cos(\lambda + \theta) 
\end{pmatrix} \;.
$$
Since $\sigma^R(x)$ has the same structure of $\sigma(x)$, the only difference being a shift of the argument of the trigonometric functions $\lambda \to \lambda+\theta$ the calculation of the modified drift is performed in the same way and the result is $\widehat{F}^R(X_t)= F^R(X_t) -\left(f\cdot f'\right) (|X_t|^2) X_t$. Therefore, all the conditions of Assumption \ref{ass:symmetry} are satisfied and the FRs \eqref{eq:diffFR1} and \eqref{eq:diffFR2} hold true for the observable 
$$A_T = \int_{0}^{T}  f^{-2}(|X_{t}|^2)\circ \Big(dX_{t} - F_S(X_t) + \frac{1}{2} \left(f\cdot f'\right) (|X_t|^2) X_t dt \Big). $$ In particular, recalling that $(J_r^\dag)^{-1}= (U_\theta^\dag)^{-1}= U_\theta$, one has the Fluctuation Relations \eqref{eq:diffFR1}-\eqref{eq:diffFR2}
\begin{equation}
\frac{d\mathcal{P}_{A_T}}{d\mathcal{P}_{U_\theta A_T}}(a)=\exp \Big\langle (U_\theta - \bbbone)\overline{F},U_\theta^\dag a \Big\rangle  \; ,
\end{equation}
\begin{equation}
\mathbb{E}\Big[\exp\Big\langle k,A_T\Big\rangle \Big]=\mathbb{E}\Big[\exp \Big\langle U_\theta^\dag k -
(\bbbone-U_\theta)\overline{F}, \,A_T\Big\rangle \Big] .
\end{equation}

\section*{Acknowledgments}

This work was funded under the Horizon Europe research and innovation programme through the MSCA project ConNEqtions, n. 101056638.  S.M.  also gratefully acknowldges financial support from the Italian Ministry of University and Research and Next Generation EU through the PRIN 2022 project ONES,  CUP:D53C24003430001,  and from the European Research Council through the ERC StG MaTCh,  grant agreement n. 101117299.
S.M. acknowledges the collaboration with Juan P. Garrahan and Carlos P\'erez-Espigares on related projects. The work of S.M. has been carried out under the auspices of the GNFM of INdAM. 
R.C. is supported by the project RETENU ANR-20-CE40 of the French National Research Agency (ANR). 
R.C. wants to pay tribute to Krzysztof Gawedzki, who left us 2 years ago. It is with him that R.C. began to be interested in Fluctuation Relations. 

\paragraph{Data Availability Statement.} No datasets were generated or analysed during the current study.

\paragraph{Conflicts of Interest.} The authors have no conflicts of interest to declare that are relevant to the content of this article.

\appendix

\section{Calculations in the semi-Markov process example}
\label{app_semimarkov}

Recall the probability density of a trajectory given in \eqref{eq:ptraj} and substitute the assumptions on the semi-Markov kernel 
\begin{equation}\label{eq:rw}
    M(x,y,t) = 
    \begin{cases}
        p M(t) \quad  \mathrm{if} \quad  y = x+1 \; \mathrm{MOD} \; L \\
        q M(t) \quad  \mathrm{if} \quad  y = x-1 \; \mathrm{MOD} \; L \\
        0 \quad\quad\quad  \mathrm{otherwise}
    \end{cases}, 
\end{equation}
where $p,q >0, p+q=1$ and $M(t)= \frac{\lambda^{n}t^{N-1}}{\left(n-1\right)!}\exp\left(-\lambda t\right)$. Explicitly, for a trajectory with $N$ jumps, one has
\begin{align}
    d \mathbb{P}_{[0,T]}^0 [X_{[0,T]}] = \pi(x_0) \prod_{i=0}^{N-1} p(x_i,x_{i+1}) M(t_{i+1}-t_i) \left(\int_{T-t_N}^{\infty} M(\tau) d\tau \right) \prod_{j=1}^{N} dt_j \;.
\end{align}

If $A_T$ is an observable counting the number of jumps of a certain type, each term $\mathrm{e}^{\left\langle k, A_T (X_{[0,T]}) \right\rangle} d \mathbb{P}_{[0,T]}^0 [X_{[0,T]}]$ will be factorized into a time-dependent term that involves the details about the jump times (and therefore the sojourn times) and a time-independent term that contains the details about the visited configurations and the different types of jumps occurred. For a fixed number of jumps, and a fixed set of jump times, one can sum over the possible sequences of visited configurations. For instance, for any even initial configuration $x_0$, one has a total of four possible two-jump trajectories with the following weights:
\begin{align}\label{eq:trees}
    &x_0  \overset{p \mathrm{e}^{k_1}}{\longrightarrow} x_0  + 1 \overset{q \mathrm{e}^{k_1}}{\longrightarrow} x_0 \,,\\
    &x_0  \overset{q \mathrm{e}^{k_2}}{\longrightarrow} x_0  - 1 \overset{p \mathrm{e}^{k_2}}{\longrightarrow} x_0 \,,\\
    &x_0  \overset{p \mathrm{e}^{k_1}}{\longrightarrow} x_0  + 1 \overset{p \mathrm{e}^{k_2}}{\longrightarrow} x_0 + 2 \,,\\
    &x_0  \overset{q \mathrm{e}^{k_2}}{\longrightarrow} x_0  - 1 \overset{q \mathrm{e}^{k_1}}{\longrightarrow} x_0 -2  \,.
\end{align}
For odd initial configurations, one finds the same expressions upon exchanging $k_1$ with $k_2$. This is because a jump $x \to x+1$ will count as a even-bond jump (odd-bond jump), and therefore keep a factor $\mathrm{e}^{k_1}$ ($\mathrm{e}^{k_2}$), if $x$ is even (odd). Summing over, one finds 
\begin{equation*}
   \sum_{x_0 \,\text{even}} \pi(x_0) \Big( pq \mathrm{e}^{2 k_1} + pq \mathrm{e}^{2 k_2} + (p^2 + q^2)\mathrm{e}^{k_1+k_2} \Big) = \frac{1}{2} \eta_{(k_1,k_2)} \;,
\end{equation*}
where we defined the function $\eta_{(k_1,k_2)} \equiv (p \mathrm{e}^{k_1} + q \mathrm{e}^{k_2})(q \mathrm{e}^{k_1} + p \mathrm{e}^{k_2})$ and we used the fact that $\sum_{x_0 \,\text{even}} \pi(x_0)=1/2$ by symmetry. The expression is symmetric under the exchange $k_1 \leftrightarrow k_2$ and therefore the contribution of odd $x_0$ reads the same. When considering the trajectories with three jumps one can continue by adding a further jump from the $N=2$ cases and each ending configuration in \eqref{eq:trees} can further jump right ($p\mathrm{e}^{k_1}$)  or left ($q\mathrm{e}^{k_2}$). When summing over, one gets an overall factor $p\mathrm{e}^{k_1} + q\mathrm{e}^{k_2}$ for the even $x_0$ case and a factor $p\mathrm{e}^{k_2} + q\mathrm{e}^{k_1}$ for odd $x_0$. Since, as we already saw, the factor $\sum_{x_0 \,\text{even}} \pi(x_0) = \sum_{x_0 \,\text{odd}} \pi(x_0)= 1/2$ is the same for the two cases, one has a global factor
\begin{equation*}
    \frac{1}{2}\eta_{(k_1,k_2)} \Big( p\mathrm{e}^{k_1} + q\mathrm{e}^{k_2} + p\mathrm{e}^{k_2} + q\mathrm{e}^{k_1} \Big) = \frac{\mathrm{e}^{k_1} + \mathrm{e}^{k_2}}{2} \, \eta_{(k_1,k_2)}.
\end{equation*}
More in general, repeating the same reasoning, one finds that for $N=2\ell$ the structural factor reads $\eta_{(k_1,k_2)}^\ell$, i.e. $\eta$ to the power $\ell$, while for $N=2\ell +1$ the factor is $\frac{\mathrm{e}^{k_1} + \mathrm{e}^{k_2}}{2}\eta_{(k_1,k_2)}^\ell$.
Therefore, one can write the moment generating function as follows 
\begin{equation}
    \mathbb{E}_{\mathbb{P}^0} \left[ \mathrm{e}^{\left\langle k, A_T \right\rangle} \right] = \sum_{\ell=0}^\infty f_{2\ell}(T) \eta_{(k_1,k_2)}^\ell + \frac{\mathrm{e}^{k_1}+\mathrm{e}^{k_2}}{2} \sum_{\ell=0}^\infty f_{2\ell+1}(T) \eta_{(k_1,k_2)}^\ell,
\end{equation}
where the function $f_N(T)$ depends on the number of jumps $N$ and reads
\begin{align}
    f_N(T)= \frac{\mathrm{e}^{-\lambda T} (\lambda T)^{nN}}{\Gamma(n)^N} \sum_{j=0}^{n-1}\frac{\lambda^j T^j}{j!} \int_0^1 \tau^{Nn-1} (1-\tau)^j d\tau \, I_{N-1}I_{N-2} \ldots I_1,
\end{align}
with $$I_j= \int_0^1 x^{jn-1}(1-x)^{n-1} dx$$ and $\Gamma$ is the Euler gamma function.
For $n=1$, that is the Markovian case, one has $I_j=1/j$, $\Gamma(n)=1$ and therefore $f_N(T)= \mathrm{e}^{-\lambda T} (\lambda T)^{N}/N!$. This allows to easily sum the series 
\begin{equation}
    (n=1): \qquad \mathbb{E}_{\mathbb{P}^0} \left[ \mathrm{e}^{\left\langle k, A_T \right\rangle} \right] = \mathrm{e}^{-\lambda T} \big[ \cosh(\lambda T \sqrt{\eta}) + \frac{\mathrm{e}^{k_1}+\mathrm{e}^{k_2}}{2\sqrt{\eta}}\sinh(\lambda T \sqrt{\eta}) \big].
\end{equation}
The non Markovian case $n=2$ is slightly more difficult but can be done similarly. One can see that $$
I_j= \frac{1}{(2j)(2j+1)} \quad \text{and} \quad \sum_{j=0}^{1}\frac{\lambda^j T^j}{j!} \int_0^1 \tau^{2N-1} (1-\tau)^j d\tau = \frac{2N + 1 + \lambda T}{2N(2N+1)}, $$
and therefore
\begin{equation}
     f_N(T)= \mathrm{e}^{-\lambda T}  \left( \frac{ (\lambda T)^{2N}}{(2N)!} + \frac{ (\lambda T)^{2N+1}}{(2N+1)!}  \right) .
\end{equation}
This in turn allows to write the moment generating function as
\begin{equation}
    (n=2): \quad \mathbb{E}_{\mathbb{P}^0} \left[ \mathrm{e}^{\left\langle k, A_T \right\rangle} \right] = \mathrm{e}^{-\lambda T} \sum_{\ell=0}^\infty \eta^\ell \big[ \frac{ (\lambda T)^{4\ell}}{(4\ell)!} + \frac{ (\lambda T)^{4\ell+1}}{(4\ell+1)!}  + \frac{\mathrm{e}^{k_1}+\mathrm{e}^{k_2}}{2} \big( \frac{ (\lambda T)^{4\ell+2}}{(4\ell+2)!} + \frac{ (\lambda T)^{4\ell+3}}{(4\ell+3)!}  \big)  \big] .
\end{equation}
This expression can also be summed recalling that
\begin{align}
    &\sum_{\ell=0}^\infty \frac{ x^{4\ell}}{(4\ell)!} = \frac{\cosh(x)+\cos(x)}{2} \qquad \qquad \sum_{\ell=0}^\infty \frac{ x^{4\ell+1}}{(4\ell+1)!} = \frac{\sinh(x)+\sin(x)}{2} \nonumber \\
    & \sum_{\ell=0}^\infty \frac{ x^{4\ell+2}}{(4\ell+2)!} = \frac{\cosh(x)-\cos(x)}{2} \qquad \; \sum_{\ell=0}^\infty \frac{ x^{4\ell+3}}{(4\ell+3)!} = \frac{\sinh(x)-\sin(x)}{2} \nonumber
\end{align}
For a generic $n$ the integral $I_j$ reads as follows
\begin{equation}
    I_j = \frac{(n-1)!}{nj(nj+1)\ldots (nj+n-1)}
\end{equation}
and one also has\footnote{These are instances of the formula $I(a,b) \equiv \int_0^1 x^{a-1} (1-x)^b dx = \frac{b!}{a(a+1)\ldots (a+b)}$, for $a,b \in \mathbb{N}, a\geq 1,b \geq 0$, that can be proved recursively using integration by parts. Indeed, one has the relation $I(a,b)= \frac{b}{a}I(a+1,b-1)$ that can be iterated until reaching $I(a+b,0)=\frac{1}{a+b}$.  }
\begin{equation}
    \sum_{j=0}^{n-1}\frac{\lambda^j T^j}{j!} \int_0^1 \tau^{nN-1} (1-\tau)^j d\tau = \frac{1}{nN} + \frac{\lambda T}{nN(nN + 1)}+ \ldots+ \frac{\lambda^{n-1}T^{n-1}}{nN(nN+1)\ldots (nN+n-1)}.
\end{equation}
Therefore, 
\begin{equation}
   f_N(T)= \mathrm{e}^{-\lambda T}\sum_{j=0}^{n-1}\frac{(\lambda T)^{n N+j}}{(n N + j)!}
\end{equation}
and the moment generating function reads
\begin{equation}
   \mathbb{E}_{\mathbb{P}^0} \left[ \mathrm{e}^{\left\langle k, A_T \right\rangle} \right] = \mathrm{e}^{-\lambda T} \sum_{\ell=0}^\infty \left(  \eta_{(k_1,k_2)}^\ell \sum_{j=0}^{n-1}\frac{(\lambda T)^{2\ell n+j}}{(2 \ell n + j)!} + \frac{\mathrm{e}^{k_1}+\mathrm{e}^{k_2}}{2}\eta_{(k_1,k_2)}^\ell  \sum_{j=0}^{n-1}\frac{(\lambda T)^{(2\ell +1) n+j}}{(2 \ell n + n + j)!} \right).
\end{equation}
The series can be summed exactly using the following formula, for $j \in \{0,1,\ldots, 2n-1\}$,
\begin{equation}
    \sum_{\ell =0}^\infty \frac{x^{2n\ell+j}}{(2n\ell +j)!} = \frac{1}{2n}\sum_{p=0}^{2n-1} \mathrm{e}^{-i\frac{\pi}{n}p j} \mathrm{e}^{x \, \mathrm{e}^{i\pi p/n}} 
\end{equation}
and the result is
\begin{equation}
    \mathbb{E}_{\mathbb{P}^0} \left[ \mathrm{e}^{\left\langle k, A_T \right\rangle} \right] = \frac{\mathrm{e}^{-\lambda T}}{2n} \sum_{p=0}^{2n -1} \mathrm{e}^{\lambda T \eta^{1/(2n)} \, \mathrm{e}^{i\pi p/n}} \sum_{j=0}^{n-1} \mathrm{e}^{-i\frac{\pi}{n}p j}  \left( \frac{1}{\eta^{\frac{j}{2n}}} + \frac{\mathrm{e}^{k_1}+\mathrm{e}^{k_2}}{2} \frac{(-1)^p}{\eta^{\frac{n+j}{2n}}}\right).
\end{equation}
Among the many time-dependent exponentials in the sum, the leading one is the term corresponding to $p=0$, so that one can compute the scaled cumulant generating function as presented in Eq.\eqref{eq:lambdaSemiMarkov}
\begin{align*}
    \Lambda_0(k) \equiv \lim_{T \to \infty} \frac{1}{T}\ln \mathbb{E}_{\mathbb{P}^0} \left[ \mathrm{e}^{\left\langle k, A_T \right\rangle} \right] &= \lambda (\eta_{(k_1,k_2)})^{\frac{1}{2n}} - \lambda \\
    &= \lambda \left( (p \mathrm{e}^{k_1} + q \mathrm{e}^{k_2})(q \mathrm{e}^{k_1} + p \mathrm{e}^{k_2}) \right)^{1/2n} - \lambda .
\end{align*}

One can also compute explicitly the rate function, that is the Legendre transform of $\Lambda_0$
\begin{equation}
    I_0(a)= \sup_{k \in \mathbb{R}^d}\Big( \left\langle k, a \right\rangle - \Lambda_0(k) \Big).
\end{equation}
Since $\Lambda_0$ is smooth, in order to find the supremum we look for stationary points of the gradient and further inspect the Hessian matrix to understand the nature of these stationary points.
The condition on the gradient for a stationary point $k^*$ reads
\begin{align*}
    &\partial_1 \Big( \left\langle k, a \right\rangle - \Lambda_0(k) \Big) = a_1 - \partial_1 \Lambda_0(k) = a_1 - \frac{\lambda}{2n} \eta^{1/2n - 1} \partial_1 \eta \overset{k=k^*}{=} 0 \,,\\
    &\partial_2 \Big( \left\langle k, a \right\rangle - \Lambda_0(k) \Big) = a_2 - \partial_2 \Lambda_0(k) = a_2 - \frac{\lambda}{2n} \eta^{1/2n - 1} \partial_2 \eta \overset{k=k^*}{=} 0 \,.
\end{align*}
From the explicit expression of $\eta$ one finds that $\partial_1 \eta + \partial_2 \eta = 2 \eta$, for any $k$, so that summing the two equations one has
\begin{equation}\label{eq:sum}
    a_1 + a_2 = \frac{\lambda}{n} \eta^{1/2n} \Big\vert_{k^*} ,
\end{equation}
and therefore also
\begin{equation}\label{eq:der}
    \partial_1 \eta \big\vert_{k^*}  = \eta \big\vert_{k^*}  \frac{2 a_1}{a_1 + a_2} , \qquad \partial_2 \eta \big\vert_{k^*}  = \eta \big\vert_{k^*}  \frac{2 a_2}{a_1 + a_2}.
\end{equation}
Recalling the explicit expression of $\eta$ and introducing the variables $x= \mathrm{e}^{k_1^*}$ and $y=\mathrm{e}^{k_2^*}$ one has from \eqref{eq:sum}
\begin{equation}
    (p^2 + q^2) xy + pq (x^2 + y^2) = \left( \frac{n}{\lambda} (a_1 + a_2) \right)^{2n}.
\end{equation}
Also, subtracting the two equations in \eqref{eq:der}, a second expression relating $x$ and $y$ reads
\begin{equation}
    pq (x^2 - y^2) = \frac{a_1-a_2}{a_1+a_2} \left( \frac{n}{\lambda} (a_1 + a_2) \right)^{2n}.
\end{equation}
These two equations are conveniently solved in terms of the variables $z=x+y$ and $w=x-y$. In particular, defining also the constant $C^2=(n(a_1+a_2)/\lambda)^{2n}$, they read
\begin{align}
    &z^2 - (p-q)^2 w^2 = 4 C^2, \\
    & (pq) z w = \left( \frac{a_1-a_2}{a_1+a_2} \right) C^2.
\end{align}
One can write $w$ in terms of $z$ from the second equation
\begin{equation}\label{eq:w}
    w= \left(\frac{a_1-a_2}{a_1+a_2} \right)  \frac{C^2}{pq}\frac{1}{z}
\end{equation}
and then substitute it in the first one, obtaining a quadratic equation for $z^2$
\begin{equation}
    z^4 - 4C^2 z^2 - \frac{(p-q)^2}{p^2 q^2}\left(\frac{a_1-a_2}{a_1+a_2} \right)^2 C^4 =0,
\end{equation}
that can be readily solved (the negative root is discarded because incompatible with $z^2$)
\begin{equation}
    z^2 = 2 C^2 \left(1 + \sqrt{1 + \frac{(p-q)^2}{4 p^2 q^2}\left(\frac{a_1-a_2}{a_1+a_2} \right)^2 } \right).
\end{equation}
Therefore,
\begin{equation}
    z= \sqrt{2}C \sqrt{1 + \sqrt{1 + \frac{(p-q)^2}{4 p^2 q^2}\left(\frac{a_1-a_2}{a_1+a_2} \right)^2 }}
\end{equation}
and one has also
\begin{equation}
    \frac{1}{z}= \frac{1}{\sqrt{2}C} \sqrt{-1 +\sqrt{1 + \frac{(p-q)^2}{4 p^2 q^2}\left(\frac{a_1-a_2}{a_1+a_2} \right)^2 } } \frac{2pq}{|p-q|}\frac{a_1+a_2}{|a_1-a_2|} ,
\end{equation}
so that from \eqref{eq:w} 
\begin{equation}
    w= \frac{\sqrt{2}C}{|p-q|} \mathrm{sign}(a_1-a_2) \sqrt{-1 +\sqrt{1 + \frac{(p-q)^2}{4 p^2 q^2}\left(\frac{a_1-a_2}{a_1+a_2} \right)^2 } } .
\end{equation}
Coming back to the original variables, one arrives at
\begin{align}
    &k_1^* = \ln(x) = \ln \left(\frac{z+w}{2}\right) \\
    &=n \ln \left(\frac{n (a_1+a_2)}{\lambda}\right)  - \frac{\ln(2)}{2} \nonumber \\
    &+ \ln\left( \sqrt{1 + \sqrt{1 + \frac{(p-q)^2}{4 p^2 q^2}\left(\frac{a_1-a_2}{a_1+a_2} \right)^2 }} + \frac{\mathrm{sign}(a_1-a_2)}{|p-q|}\sqrt{-1 +\sqrt{1 + \frac{(p-q)^2}{4 p^2 q^2}\left(\frac{a_1-a_2}{a_1+a_2} \right)^2 } } \right)
\end{align}
\begin{align}
    &k_2^* = \ln(y) = \ln \left(\frac{z-w}{2}\right) \\
    &=n \ln \left(\frac{n (a_1+a_2)}{\lambda}\right)  - \frac{\ln(2)}{2} \nonumber  \\
    &+ \ln\left( \sqrt{1 + \sqrt{1 + \frac{(p-q)^2}{4 p^2 q^2}\left(\frac{a_1-a_2}{a_1+a_2} \right)^2 }} - \frac{\mathrm{sign}(a_1-a_2)}{|p-q|}\sqrt{-1 +\sqrt{1 + \frac{(p-q)^2}{4 p^2 q^2}\left(\frac{a_1-a_2}{a_1+a_2} \right)^2 } } \right)
\end{align}

\section{Calculations in the Langevin equation example}
\label{app_langevin}

Substituting in the equation \eqref{eq:Ric} the following ansatz for $B_k$, where $a,c \in \mathbb{R}$ and $b \in \mathbb{C}$,
\begin{equation*}
    B_k = 
    \begin{pmatrix}
        a & b \\
        b^* & c
    \end{pmatrix} \otimes \bbbone_2
\end{equation*}
one obtains a system of equations for the triple $(a,b,c)$
\begin{equation}\label{eq:abc}
   \begin{cases}
       \frac{4\gamma}{\beta} (a^2 + |b|^2) -2\gamma a + k_1 =0 \\
       \frac{4\gamma}{\beta} (c^2 + |b|^2) -2\gamma c + k_2 =0 \\
       \frac{4\gamma}{\beta} b (a + c) -2\gamma b + \frac{k_3}{2} =0
   \end{cases} 
\end{equation}
and one can check that for $k=0$ the physical solution $(a,b,c)=(0,0,0)$ is allowed, among others. When $k$ in finite we look for a solution that goes to zero in the limit $k \to 0$. Defining the real variables $x= a+c -\beta/2$ and $y= a-c$, and subtracting the second equation from the first one, one finds
\begin{equation}\label{eq:xy}
    xy = - \frac{\beta}{4\gamma} (k_1 - k_2).
\end{equation}
Also, the last equation can be rewritten as
\begin{equation}\label{eq:b}
bx=- \beta k_3/(8\gamma)  \,,
\end{equation}
so that $b$ has to be real ($b=b^*$) whenever $x\neq 0$.
Combining \eqref{eq:b} with \eqref{eq:xy} one also gets
\begin{equation}\label{eq:ky}
    k_3 y = 2 (k_1 -k_2) b .
\end{equation}
Moreover, one can isolate $b^2$ from the first equation in \eqref{eq:abc}, 
\begin{equation*}
    b^2 = - \frac{x^2 + y^2}{4} -\frac{xy}{2} +\frac{\beta^2}{16} - \frac{\beta k_1}{4\gamma} ,
\end{equation*}
multiply both members by $4 x^2 (k_1 - k_2)^2$ and use \eqref{eq:ky} on the left hand side
\begin{equation*}
    x^2 k_3^2  y^2  = \left( - \frac{x^2 +y^2}{4} -\frac{xy}{2} +\frac{\beta^2}{16} - \frac{\beta k_1}{4\gamma} \right) 4 x^2 (k_1 -k_2)^2 .
\end{equation*}
Finally, one uses \eqref{eq:xy} to substitute the terms $xy$ and $x^2y^2$, 
\begin{equation*}
    \frac{\beta^2}{16\gamma^2} (k_1 -k_2)^2 \big[ k_3^2 + (k_1 -k_2)^2 \big] =\left[ - x^4  + \left( -\frac{\beta}{2\gamma}(k_1 +k_2)  + \frac{\beta^2}{4} \right) x^2 \right] (k_1 -k_2)^2 
\end{equation*}
so that for $(k_1 - k_2) \neq 0$ the following biquadratic equation for $x$ holds true
\begin{equation}
    x^4 - 2 v_k  x^2 + w_k =0, \quad v_k = \left( \frac{\beta^2}{8} -\frac{\beta}{4\gamma}(k_1 +k_2)   \right) , \quad w_k = \frac{\beta^2}{16\gamma^2}\big[ k_3^2 + (k_1 -k_2)^2 \big],
\end{equation}
where the variables $v_k$ and $w_k$ have been defined for convenience. We notice that $w_k \geq 0$ while $v_k$ can be positive or negative. The solution $x$ can be a real number only if $v_k \geq 0$, that is
\begin{equation}\label{eq:bound1}
    k_1+k_2 \leq \frac{\gamma\beta}{2}.
\end{equation}
We choose the solution for $x^2$ that goes to $\beta^2/4$ for $k=0$ (remember that $x= a+c -\beta/2$)
\begin{equation}
    x^2 = v_k + \sqrt{v_k^2 - w_k} .
\end{equation}
The condition $x^2 \geq 0$ necessary to get a real $x$ imposes the bound $v_k^2- w_k \geq 0$, that is
\begin{equation}\label{eq:bound2}
   k_3^2 - 4 k_1k_2   \leq \frac{\beta^2\gamma^2}{4} -  \beta \gamma(k_1+k_2).
\end{equation}
Finally, we consider the square root with a minus sign because $x$ is negative $(-\beta/2)$ when $k=0$, namely
\begin{equation}
    x= - \sqrt{  v_k + \sqrt{v_k^2 - w_k}}.
\end{equation}
As a consequence, from Eq. \eqref{eq:xy} we find
\begin{equation}
    y= \frac{\beta}{4\gamma} \frac{(k_1 - k_2)}{\sqrt{w_k}}\sqrt{  v_k - \sqrt{v_k^2 - w_k}}, 
\end{equation}
that even though $w_k \to 0$ for vanishing $k$ has a well defined limit and can be extended by continuity. In fact, one has $y=0$ for $k=0$. Similarly, it turns out that
\begin{equation}\label{eq:b2}
    b= \frac{\beta}{8\gamma} \frac{k_3}{\sqrt{w_k}}\sqrt{  v_k - \sqrt{v_k^2 - w_k}}.
\end{equation}
Therefore, one can write for the scaled cumulant generating function for $k$ in a set $S \subset \mathbb{R}^3$ defined by the two conditions \eqref{eq:bound1}-\eqref{eq:bound2}
\begin{align}
    \Lambda_0(k) &= \mathrm{Tr}(D B_k) = \frac{2\gamma}{\beta} (a+c)= \frac{2\gamma}{\beta}\left( x+ \frac{\beta}{2} \right) \nonumber\\
    &= \gamma - \frac{2\gamma}{\beta}\sqrt{  v_k + \sqrt{v_k^2 - w_k}} .
\end{align}

We now check that the modified drift $M_k := M - 2D B_k$ has at least one eigenvalue with non-positive real part for $k \notin S$, thus suggesting that $\Lambda_0(k)$ is infinite for those values of $k$.
Explicitly,
\begin{align}
    &\det(M_k - \lambda \bbbone_4) = 
    \det 
    \begin{pmatrix}
     -\lambda & -1 & 0 & 0 \\
     1 & \gamma - \frac{4 \gamma}{\beta}a -\lambda & 0 & - \frac{4 \gamma}{\beta}b \\
     0 & 0 & -\lambda & -1 \\
     0 & - \frac{4 \gamma}{\beta}b & 1 & \gamma - \frac{4 \gamma}{\beta}c - \lambda
    \end{pmatrix} \nonumber \\
    &= \left( 1 + \lambda^2 + \frac{2\lambda\gamma}{\beta}\left(a +c - \frac{\beta}{2}\right) \right)^2  - \frac{4\lambda^2 \gamma^2}{\beta^2}\Big( (c-a)^2 + 4b^2 \Big) \\
    &= \Big[ 1 + \lambda^2 + \frac{2\lambda \gamma}{\beta}\left(a+c - \frac{\beta}{2} - b \frac{8\gamma}{\beta} \frac{\sqrt{w_k}}{k_3}\right) \Big] \Big[ 1 + \lambda^2 + \frac{2\lambda \gamma}{\beta}\left(a+c - \frac{\beta}{2} + b \frac{8\gamma}{\beta} \frac{\sqrt{w_k}}{k_3}\right)  \Big] ,
\end{align}
where in the last step we used 
\begin{equation*}
    (c-a)^2 + 4b^2 = y^2 + 4b^2 = \frac{4b^2}{k_3^2} \Big(k_3^2 + (k_1 - k_2)^2 \Big) = \frac{64 \gamma^2}{\beta^2}\frac{w_k}{k_3^2} b^2 .
\end{equation*}
Recalling the explicit expression of $x=a+c-\beta/2$ and using \eqref{eq:b2} we find
\begin{align}
    \lambda_{1,2}= \frac{\gamma}{\beta}\left( V_+ + V_- \right) \pm \sqrt{\frac{\gamma^2}{\beta^2}\left(V_+ + V_- \right)^2 -1} \\
    \lambda_{3,4}= \frac{\gamma}{\beta}\left( V_+ - V_- \right) \pm \sqrt{\frac{\gamma^2}{\beta^2}\left(V_+ - V_- \right)^2 -1}
\end{align}
where we defined for convenience 
\begin{equation}
    V_+ := \sqrt{  v_k + \sqrt{v_k^2 - w_k}}, \qquad V_- := \sqrt{  v_k - \sqrt{v_k^2 - w_k}} .
\end{equation}
If $v_k > 0$ and $v_k^2 - w_k > 0$, in the interior of $S$, then $V_+ > V_- \geq 0$ and all the $\lambda_i$ have a positive real part (because the term out of the square root is positive, and the square root is either purely imaginary or it is real but smaller than the first term in absolute value). Conversely, if $v_k^2 - w_k <0$, then $V_+$ and $V_-$ are complex conjugates and both $\lambda_3$ and $\lambda_4$ are purely imaginary. Also, when $v^2-w_k>0$ but $v_k <0$, $V_+$ and $V_-$ are imaginary and therefore both $\lambda_3$ and $\lambda_4$ are again purely imaginary.

We now show that the function $\Lambda_0(k)$ is steep at the boundary of $S$. The gradient in the interior of $S$ reads explicitly
\begin{align}
    \partial_1 \Lambda_0(k) &= \frac{1}{4\sqrt{  v_k + \sqrt{v_k^2 - w_k}}}  \left( 1 - \frac{\beta}{2\gamma\sqrt{v_k^2 - w_k}} \left( k_2 - \frac{\beta \gamma}{4} \right) \right) \label{eq:partial1}\\
    \partial_2 \Lambda_0(k) &= \frac{1}{4\sqrt{  v_k + \sqrt{v_k^2 - w_k}}}  \left( 1- \frac{\beta}{2\gamma\sqrt{v_k^2 - w_k}} \left( k_1 - \frac{\beta \gamma}{4} \right) \right) \label{eq:partial2}\\
    \partial_3 \Lambda_0(k) &=\frac{1}{4\sqrt{  v_k + \sqrt{v_k^2 - w_k}}}  \left( \frac{\beta}{4\gamma \sqrt{v_k^2 - w_k}}  k_3   \right), \label{eq:partial3}
\end{align}
where we used $\partial_1 v_k= \partial_2 v_k = -\beta/(4\gamma)$ , $\partial_1 w_k = - \partial_2 w_k = (k_1-k_2)\beta^2/(8\gamma^2)$, $\partial_3 v_k=0$, $\partial_3 w_k = k_3 \beta^2/(8\gamma^2)$, that imply also $$2v_k \partial_1 v_k - \partial_1 w_k= \frac{\beta^2}{4\gamma^2}\left( k_2 - \frac{\beta\gamma}{4} \right), \qquad 2v_k \partial_2 v_k - \partial_2 w_k= \frac{\beta^2}{4\gamma^2}\left( k_1 - \frac{\beta\gamma}{4} \right) .$$
Therefore
\begin{align}
    \big|\nabla \Lambda_0(k)\big|^2 &= \frac{1}{16\left(v_k + \sqrt{v_k^2 - w_k} \right)\Big( v_k^2 -w_k \Big)}\Bigg[ 2(v_k^2 -w_k ) - \frac{\beta}{\gamma}\sqrt{v_k^2 - w_k} \left( k_1+ k_2 - \frac{\beta\gamma}{2} \right) \nonumber \\
    &\quad + \frac{\beta^2}{4\gamma^2}\left( k_1^2 +k_2^2 -\frac{\beta\gamma}{2}(k_1+k_2) + \frac{\beta^2\gamma^2}{8} + \frac{k_3^2}{4} \right)\Bigg] \nonumber \\
    &= \frac{1}{16\left(v_k + \sqrt{v_k^2 - w_k} \right)\Big( v_k^2 -w_k \Big)}\Bigg[ 2(v_k^2 -w_k ) + 4 v_k \sqrt{v_k^2 - w_k} + 2 v_k^2 \nonumber \\
    &\quad +\frac{\beta^2}{16\gamma^2}\Big(k_3^2 + 2(k_1-k_2)^2\Big) \Bigg]
\end{align}
where in the second equality we used 
\begin{equation}
    \frac{\beta^2\gamma^2}{8} - \frac{\beta\gamma}{2}(k_1 + k_2) = \frac{1}{2}\left( \frac{\beta\gamma}{2} - (k_1+k_2) \right)^2 - \frac{(k_1+k_2)^2}{2} = \frac{8 \gamma^2}{\beta^2}v_k^2 -  \frac{(k_1+k_2)^2}{2} .
\end{equation}
Since one has
\begin{equation}
 -2w_k  \leq \frac{\beta^2}{16\gamma^2}\Big(k_3^2 + 2(k_1-k_2)^2\Big) \leq 2w_k
\end{equation}
it turns out that
\begin{equation}
  \frac{1}{4 \sqrt{v_k^2 -w_k }}  \leq  \big|\nabla \Lambda_0(k)\big|^2 \leq \frac{v_k}{4 (v_k^2 -w_k )} = \frac{1}{4 \sqrt{v_k^2 -w_k }} \left( \frac{v_k}{\sqrt{v_k^2 -w_k}} \right) .
\end{equation}
In particular, the lower bound implies that the function is steep for $k$ converging to the boundary of $S$, $v_k^2-\omega_k=0$.

We finally compute the rate function $I_0(a)$. We look for a stationary point $k^* \in S$ of the function $\langle a, k \rangle - \Lambda_0(k)$. The stationary point condition amounts to equate the three expressions \eqref{eq:partial1}-\eqref{eq:partial2}-\eqref{eq:partial3} to $a_1,a_2$ and $a_3$, respectively. This in turn implies the following three equations
\begin{align}
    \left.\frac{1}{8\gamma \sqrt{  v_k + \sqrt{v_k^2 - w_k}} }\right\vert_{k^*} \left. \frac{\beta}{\sqrt{v_k^2 - w_k}}\right\vert_{k^*} (k^*_1 - k^*_2) &= a_1 -a_2 \label{eq:firsteq}\\
     - \left.\frac{\beta}{2\gamma\sqrt{v_k^2 - w_k}}\right\vert_{k^*}\left( a_2 k_2^* - a_1 k_1^* - \frac{\beta\gamma}{4}(a_2 - a_1) \right) &= a_1 - a_2 \label{eq:secondeq}\\
     \left.\frac{1}{16\gamma \sqrt{  v_k + \sqrt{v_k^2 - w_k}} }\right\vert_{k^*} \left. \frac{\beta}{\sqrt{v_k^2 - w_k}}\right\vert_{k^*} k_3^* &= a_3
\end{align}
The first one is \eqref{eq:partial1}$-$\eqref{eq:partial2} $=a_1 - a_2$, while the second one is $a_2\times$ \eqref{eq:partial1} $-$ $a_1\times$ \eqref{eq:partial2}$=0$. If $a_1 \neq a_2$, we find immediately from the first and the third equation that
\begin{equation}\label{eq:k3intermediate}
    k_3^* = (k_1^*-k_2^*) \frac{2a_3}{a_1-a_2}.
\end{equation}
Using the second equation we find that
\begin{align}
    \beta^2 \left( a_2 k_2^* - a_1 k_1^* + \frac{\beta\gamma}{4}(a_1 - a_2) \right)^2 &= 4\gamma^2 (a_1-a_2)^2 \left( v_k^2 - w_k \right) \Big\vert_{k^*}, \nonumber\\
    \beta^2 \Bigg( (a_1+a_2) \frac{k_2^*-k_1^*}{2} - (a_1-a_2)\underbrace{\left(\frac{k_1^* +k_2^*}{2}  - \frac{\beta\gamma}{4}\right)}_{-2 v_k \gamma/\beta} \Bigg)^2 &= 4\gamma^2 (a_1-a_2)^2 \left( v_k^2 - w_k \right) \Big\vert_{k^*}, \nonumber \\
    \beta^2 \frac{(k_1^*-k_2^*)^2}{4}(a_1+a_2)^2 + 2 \beta\gamma v_{k^*} (k_2^*-k_1^*)(a_1^2-a_2^2) &= - 4\gamma^2 (a_1-a_2)^2 w_{k^*}
\end{align}
and substituting the expression for $k_3^*$ in $w_{k^*}$ the equation reads
\begin{equation*}
    (k_1^*-k_2^*)\left(  \frac{(k_1^*-k_2^*)}{4} \Big((a_1+a_2)^2 + 4a_3^2 + (a_1-a_2)^2\Big) +  (a_1^2-a_2^2)\left(\frac{k_1^* +k_2^*}{2}  - \frac{\beta\gamma}{4}\right)  \right)= 0.
\end{equation*}
Since we are dealing with the case $a_1 \neq a_2$, and therefore from \eqref{eq:firsteq} $k_1^* \neq k_2^*$, we can simplify the factor $(k_1^*-k_2^*)$ and we obtain
\begin{equation}\label{eq:k2intermediate}
    k_2^* = \frac{1}{a_2^2+a_3^2}\Big( (a_1^2 + a_3^2)k_1^* -\frac{\beta\gamma}{4}(a_1^2-a_2^2) \Big) .
\end{equation}
Using now \eqref{eq:firsteq} together with the following identity
\begin{equation*}
    \frac{1}{\sqrt{  v_k + \sqrt{v_k^2 - w_k}} } = \frac{1}{\sqrt{w_k}} \sqrt{  v_k - \sqrt{v_k^2 - w_k}} ,
\end{equation*}
upon squaring both members we end up with
\begin{equation*}
   \left( v_{k^*} - \sqrt{v_k^2-w_k}\big\vert_{k^*} \right)(k_1^*-k_2^*)^2 =  w_{k^*}\frac{64 \gamma^2}{\beta^2} (a_1-a_2)^2 (v_{k^*}^2- w_{k^*}) 
\end{equation*}
and using \eqref{eq:secondeq} to substitute $\sqrt{v_k^2-w_k}\big\vert_{k^*}$ we can write
\begin{equation*}
  \left( v_{k^*} + \frac{\beta}{2\gamma} \frac{a_2 k_2^* - a_1 k_1^*}{(a_1-a_2)} + \frac{\beta^2}{8} \right)(k_1^*-k_2^*)^2  = 16 w_{k^*} \left( a_2 k_2^* - a_1 k_1^* - \frac{\beta\gamma}{4}(a_2 - a_1) \right)^2.
\end{equation*}
Using the definitions of $v_k$ and $w_k$ to complete the substitution and writing $k_2^*$ and $k_3^*$ in terms of $k_1^*$, after some algebra we get
\begin{equation}
    k_1^* = \frac{\gamma}{4} \left( \beta - \frac{1}{\beta}\frac{a_2^2+a_3^2}{(a_3^2 - a_1 a_2)^2} \right) .
\end{equation}
Together with \eqref{eq:k2intermediate} and \eqref{eq:k3intermediate} we also find
\begin{align}
    k_2^* &= \frac{\gamma}{4} \left( \beta - \frac{1}{\beta}\frac{a_1^2+a_3^2}{(a_3^2 - a_1 a_2)^2} \right) , \\
    k_3^* &= \frac{\gamma}{2\beta} \frac{a_3 (a_1+a_2) }{(a_3^2 - a_1 a_2)^2} . 
\end{align}
For convenience, we write explicitly also the following expressions that result from the previous ones
\begin{equation*}
    v_{k^*}= \frac{1}{16} \frac{a_1^2+a_2^2 + 2 a_3^2}{(a_3^2 - a_1 a_2)^2}, \quad w_{k^*}= \frac{1}{16^2} \frac{(a_1+a_2)^2(4a_3^2 + (a_1-a_2)^2)}{(a_3^2 - a_1 a_2)^4} ,
\end{equation*}
\begin{equation*}
    \sqrt{v_k^2-w_k}\Big\vert_{k^*} = \frac{1}{8(a_1a_2 - a_3^2)}, \quad \sqrt{v_k + \sqrt{v_k^2-w_k}} \Bigg\vert_{k^*} = \frac{a_1+a_2}{4(a_1a_2 - a_3^2)}.
\end{equation*}
Therefore, in the end, coming back to the rate function we get the expression in \eqref{eq:Ilangevin}
\begin{align}
    I_0(a) &= \sup_{k \in \mathbb{R}^n}\Big( \langle a, k \rangle - \Lambda_0(k) \Big) = \sup_{k \in S}\Big( \langle a, k \rangle - \Lambda_0(k) \Big) = \langle a, k^* \rangle - \Lambda_0(k^*) \nonumber\\
    &= \frac{\gamma\beta}{4}(a_1+a_2) - \frac{\gamma}{4\beta}\frac{(a_1+a_2)(a_1a_2- a_3^2)}{(a_3^2 - a_1 a_2)^2} -\gamma + \frac{\gamma}{4\beta} \frac{2(a_1+a_2)}{(a_1 a_2-a_3^2)} \nonumber \\
    &=-\gamma + \frac{\gamma \beta}{4}(a_1+a_2) + \frac{\gamma}{4\beta} \frac{(a_1+a_2)}{(a_1 a_2 - a_3^2)} . 
\end{align}

\end{document}